\documentclass[twoside,letterpaper]{article}

\usepackage[preprint]{aistats2027}
\usepackage{natbib}

\usepackage{amsmath,amsfonts,bm}

\def\eqref#1{equation~\ref{#1}}

\def\1{\bm{1}}

\def\va{{\bm{a}}}

\DeclareMathAlphabet{\mathsfit}{\encodingdefault}{\sfdefault}{m}{sl}
\SetMathAlphabet{\mathsfit}{bold}{\encodingdefault}{\sfdefault}{bx}{n}

\def\gA{{\mathcal{A}}}

\def\gE{{\mathcal{E}}}

\def\gG{{\mathcal{G}}}
\def\gH{{\mathcal{H}}}

\def\gN{{\mathcal{N}}}

\def\gS{{\mathcal{S}}}

\def\sF{{\mathbb{F}}}

\def\sK{{\mathbb{K}}}

\def\sP{{\mathbb{P}}}

\newcommand{\KL}{D_{\mathrm{KL}}}

\usepackage[utf8]{inputenc} 
\usepackage[T1]{fontenc}    
\usepackage{hyperref}       
\usepackage{url}            
\usepackage{booktabs}       
\usepackage{amsfonts}       
\usepackage{nicefrac}       
\usepackage{microtype}      
\usepackage{amsthm}       
\usepackage{cleveref}       
\usepackage{graphicx}
\usepackage{wrapfig}
\usepackage{enumitem}
\usepackage{doi}
\usepackage[ruled]{algorithm2e}

\usepackage{amssymb}

\newtheoremstyle{boldstatementnote}
  {\topsep}{\topsep}{\itshape}{}{\bfseries}{.}{5pt plus 1pt minus 1pt}
  {\thmname{#1}\thmnumber{ #2}\thmnote{ \textbf{(#3)}}}
\theoremstyle{boldstatementnote}
\newtheorem{theorem}{Theorem}[section]
\newtheorem{lemma}[theorem]{Lemma}
\newtheorem{proposition}[theorem]{Proposition}

\newtheorem{definition}{Definition}[section]
\newtheorem{assumption}{Assumption}[section]

\theoremstyle{plain}

\begin{document}
\runningtitle{Last-Iterate Convergence of Policy Dynamics in Zero-Sum Networked Separable Markov Games}
\runningauthor{Zailin Ma}
\twocolumn[
\aistatstitle{Last-Iterate Convergence of Policy Dynamics in Zero-Sum Networked Separable Markov Games}
\aistatsauthor{Zailin Ma\textsuperscript{*}}
\aistatsaddress{School of Mathematical Sciences\\ Peking University}
]

\begingroup
\renewcommand{\thefootnote}{\fnsymbol{footnote}}
\footnotetext[1]{Email: \href{mailto:mazailin@stu.pku.edu.cn}{mazailin@stu.pku.edu.cn}}
\endgroup

\begin{abstract}
Solving Nash equilibria for general multi-player Markov games is computationally intractable, while two-player zero-sum Markov games admit fast last-iterate policy-optimization methods. Zero-sum networked separable Markov games occupy an important middle ground: they retain global multi-player competition structure through pairwise interactions, while preserving computational tractability of Nash equilibria (NE) in the finite-horizon setting. Existing algorithms for this class either proceed through equilibrium-collapse arguments for a simplified setting where a single controller determines the transition probability, or backward dynamic programming that relies on equilibrium solvers at each stage. However, the design and analysis of direct policy-update approaches remain inadequate. To address this issue, we propose the entropy-regularized optimistic multiplicative weights update (ER-OMWU), a complementary single-loop policy dynamic that updates players' policies symmetrically and returns an approximate NE in the last iteration. We provide the first last-iterate convergence analysis of policy dynamics in the games of interest: after $\widetilde O\left(1/{\epsilon}\right)$ iterations, the returned policy is an $\epsilon$-approximate Nash equilibrium. The result preserves the near-linear convergence rate achieved by policy optimization in two-player zero-sum Markov games, but extends the policy-dynamics viewpoint to a more complicated but structured multi-player setting.
\end{abstract}

\section{Introduction}

Multi-agent reinforcement learning studies sequential decision making in environments where agents jointly affect rewards and state transitions. Nash equilibrium is a central solution concept for such strategic systems, and equilibrium-computation methods have played an important role in benchmark problems in artificial intelligence and game playing \citep{bowling2015heads,silver2017mastering,moravvcik2017deepstack,brown2018superhuman,brown2019superhuman,brown2020combining,perolat2022mastering}. The main obstacle is that general multi-player Markov games inherit the computational difficulty of multi-player normal-form games: computing even approximate Nash equilibria is intractable in general \citep{daskalakis2009complexity,chen2009settling,etessami2010complexity,rubinstein2017settling}. This motivates the search for structured competitive Markov games that retain meaningful multi-player interactions while permitting efficient equilibrium computation.

Zero-sum networked separable Markov games provide one such structured class. They model multi-player competition through local pairwise interactions while preserving an aggregate zero-sum structure, extending the zero-sum networked separable games studied in the normal-form setting \citep{cai2011minmax,cai2016zero,ao2022asynchronous}. The Markovian setting admits useful equilibrium structure: in the finite-horizon setting, approximate Nash equilibria can be computed by value iterations which solve the zero-sum networked separable normal-form games at every state and stage \citep{park2023multi,kalogiannis2023zero}. In particular, these procedures achieve near-linear convergence rate and provide a baseline for solving Nash equilibria. However, value iteration acts more as a centralized solver rather than a natural learning process. To understand how independent players adaptively reach equilibria via iterative, local policy updates over the network topology, the paradigm of learning equilibria through policy optimization remains underexplored.

Specifically, when a zero-sum networked separable Markov game has only two players, it degenerates to a two-player zero-sum Markov game. For this case, a line of work shows that policy optimization can achieve fast last-iterate convergence. Recent methods based on optimistic and entropy-regularized updates obtain finite-time equilibrium guarantees in both discounted and episodic settings \citep{wei2021last,zhang2022policy,erez2023regret,cen2022faster,yang2023ot}. In particular, entropy-regularized OMWU attains convergence rate of near-linear dependence on $1/\epsilon$ towards an $\epsilon$-approximate Nash equilibrium \citep{cen2022faster}. These results establish that policy optimization can be both direct and fast in two-player competition, but do not address whether the same behavior persists in the generalized multi-player zero-sum networked separable case.

This paper studies that question in zero-sum networked separable Markov games. The answer in the infinite-horizon setting is negative: computing stationary equilibria in this model class is computationally hard in general \citep{park2023multi}. Therefore, we ask whether a policy-optimization method can attain a last-iterate approximate Nash equilibrium in finite-horizon zero-sum networked separable Markov games with the same near-linear dependence on $1/\epsilon$ as in two-player zero-sum Markov games. 

We answer this question affirmatively. We propose the Entropy-Regularized Optimistic Multiplicative Weights Update (ER-OMWU), a single-loop policy dynamic that maintains value estimates and conducts optimistic policy updates symmetrically among players. 
We prove that the last iterate of ER-OMWU is an $\epsilon$-approximate Nash equilibrium after
\(
\widetilde O\!\left({N^2H^3\log A}/{\epsilon}\right)
\)
iterations, where $N$ is the number of players, $H$ is the horizon of the game and $A$ is the number of actions of each player. Thus, policy optimization retains the near-linear $1/\epsilon$ accuracy dependence previously available through value-iteration-based methods, while extending fast last-iterate convergence result beyond two-player zero-sum Markov games. 

Technically, we decompose the Nash-equilibrium gap into three terms including the quantal response equilirium gap, the square root of the policy divergence and an entropy regularization term. We estimate the first two terms via a coupled recursion analyses of the policy divergence and the $Q$-function errors and obtain exponential bounds over the errors, which eventually lead to the near-linear convergence rate. Our contributions are summarized as follows.
\begin{itemize}[leftmargin=*]
    \item {We propose a policy optimization algorithm, ER-OMWU, for solving NE in zero-sum networked separable Markov games.} This algorithm maintains edge-wise value estimates and performs symmetric policy updates among players along a single policy trajectory.

    \item We prove that, with the entropy temperature chosen at the target accuracy scale, the last iterate of ER-OMWU is an $\epsilon$-approximate Nash equilibrium after
    \(
        \widetilde O\!\left({N^2H^3\log A}/{\epsilon}\right)
    \)
    iterations. Hence, the convergence rate of policy optimization in the games of interests achieves the same near-linear dependence on $1/\epsilon$ as in the case of two-player zero-sum Markov games. 

    \item We evaluate ER-OMWU on finite-horizon zero-sum networked separable Markov games and compare it with a value-iteration baseline equipped with inner OMWU stage-game solvers. The experiments show that ER-OMWU reaches the targeted Nash-equilibrium gaps under the matching update budgets. Moreover, as the entropy temperature is decreased towards smaller equilibrium gaps, the observed convergence iteration follows the near-linear $1/\epsilon$ trend predicted by our theory.
\end{itemize}

\subsection{Related work}

\paragraph{Optimistic policy learning in normal-form games and Markov games.}
Optimistic and regularized learning dynamics are central tools for fast equilibrium computation in games. In normal-form games, early fast rates for two-player zero-sum games \citep{daskalakis2011near,rakhlin2013optimization} were later extended to multi-player general-sum settings through optimistic regularized learning and higher-order smoothness arguments \citep{syrgkanis2015fast,chen2020hedging,daskalakis2021near,anagnostides2022near,anagnostides2022uncoupled}. In Markov games, \citet{wei2021last,zhang2022policy,erez2023regret} extend these ideas to sequential decision making, while \citet{cen2022faster} prove fast last-iterate convergence for entropy-regularized OMWU in two-player zero-sum Markov games and \citet{yang2023ot} obtain an $O(1/T)$ averaged-policy rate. Recent work also gives near-$1/T$ guarantees for correlated or coarse correlated equilibria in full-information general-sum Markov games \citep{cai2024near,mao2024widetilde}, and payoff-based last-iterate learning has been studied in two-player zero-sum stochastic games \citep{chen2024last}. Our analysis follows the entropy-regularized last-iterate viewpoint of \citet{cen2022faster}, but the multi-player networked structure requires edge-wise $Q$-function update rather than an anti-symmetric two-player value update. Recent lower bounds for broad non-forgetful learning rules \citep{cai2024fast} further emphasize that such fast last-iterate guarantees rely on specific structure and dynamics.

\paragraph{Zero-sum networked separable Markov games.}
Zero-sum networked separable (or equivalently, polymatrix) games model multi-player competition through pairwise interactions while preserving an aggregate zero-sum structure. Their minmax and equilibrium properties are studied in \citep{cai2011minmax,cai2016zero}, and learning dynamics for normal-form zero-sum networked separable games are analyzed by \citet{ao2022asynchronous,NEURIPS2025_42fe68e7}. The Markovian counterpart has recently been developed by \citet{park2023multi} and \citet{kalogiannis2023zero}, who identify equilibrium-collapse phenomena and give efficient equilibrium-computation procedures based on dynamic programming and induced stage-game solvers. These works establish the tractability of finite-horizon zero-sum networked separable Markov games, whereas our contribution is to show that a single-loop policy dynamics can operate directly on this structure, avoiding repeated equilibrium solver calls at state-stage pairs.

\paragraph{Other structured multi-player Markov games.}
Our work is also related to a broader effort to identify structured multi-player Markov games where Nash-equilibrium learning is tractable. \citet{anagnostides2024optimistic} study optimistic policy gradient for single-controller multi-player Markov games with an equilibrium-collapse property, obtaining convergence to stationary approximate Nash equilibria. \citet{kalogiannis2024learning} analyze adversarial team Markov games, and recent convex Markov-game frameworks provide another route to tractable multi-agent learning \citep{gemp2024convex,kalogiannis2025solving}. These settings are complementary to ours: we focus on finite-horizon zero-sum networked separable Markov games with controller-separable transitions and prove a last-iterate guarantee for an OMWU-style policy dynamics.

\section{Problem formulation}
\label{sec:prelim}

\subsection{Finite-horizon zero-sum networked separable Markov games}
Consider players $\gN=[N]$ with finite action spaces $(\gA_i)_{i\in\gN}$. Let $\gA=\prod_{i=1}^{N}\gA_i$ and $A:=\max_{i\in[N]}|\gA_i|$. We first introduce the normal-form zero-sum networked separable games.
\begin{definition}[Zero-sum networked separable games]
A networked separable normal-form game is specified by an undirected connected graph $\gG=(\gN,\gE_r)$ with no self-loops and with edge-wise utilities $u_{i,j}:\gA_i\times\gA_j\to\mathbb R$. We represent each undirected edge by both ordered pairs, so $(i,j)\in\gE_r$ if and only if $(j,i)\in\gE_r$. With $\gN_i:=\{j:(i,j)\in\gE_r\}$, player $i$ has utility
\[
    u_i(\va):=\sum_{j\in\gN_i}u_{i,j}(a_i,a_j).
\]
The game is zero-sum if $\sum_{i=1}^{N}u_i(\va)=0$ for every $\va\in\gA$. 
\end{definition}
This condition concerns the aggregate utility and does not require each edge game to be zero-sum. Utilities on mixed profiles are defined by expectation under the product distribution.

We next introduce the finite-horizon Markovian setting. Let $H$ be the horizon, let $\gS$ be a finite nonempty state space with $S:=|\gS|$, and let $r_{h,i}:\gS\times\gA\to[-N,N]$ be the reward functions and $\sP_h:\gS\times\gA\to\Delta(\gS)$ be the transition kernels for $h\in[H]$. A Markov policy of player $i$ is $\pi_i=(\pi_{h,i})_{h\in[H]}$, where $\pi_{h,i}(\cdot|s)\in\Delta(\gA_i)$ for every $s\in\gS$, and a policy profile is $\pi=(\pi_1,\ldots,\pi_N)$. We slightly abuse the notation and denote $\pi_h(s)=\prod_i\pi_{h,i}(\cdot|s)$ when the subscript is the stage $h$. For a function $f$ of the state and joint action, we write $f(s,\pi_h(s))$ for its expectation under $\prod_i\pi_{h,i}(\cdot|s)$. Likewise, $\pi_{h,-i}(s)$ denotes the product distribution of the policies of all players except player $i$, and $f(s,a_i,\pi_{h,-i}(s))$ denotes the corresponding expectation.

Given a policy profile $\pi$, set $V_{H+1,i}^{\pi}\equiv0$ and define recursively
\[
    Q_{h,i}^{\pi}(s,\va)
    :=r_{h,i}(s,\va)+\sP_h(\cdot|s,\va)V_{h+1,i}^{\pi}(\cdot),
\]
\[
    V_{h,i}^{\pi}(s):=Q_{h,i}^{\pi}(s,\pi_h(s)).
\]
Here the product of a transition measure and a state function denotes their finite-state inner product. Backward induction gives
\[
    \|V_{h,i}^{\pi}\|_\infty,\ \|Q_{h,i}^{\pi}\|_\infty
    \leq N(H-h+1)\leq NH.
\]

We adopt the definition of the zero-sum networked separable Markov games that describes the decomposable reward and transition characterization in \citep{park2023multi} and restate it here.
\begin{definition}[Zero-sum networked separable Markov games]
    \label{formaldef}
    Let $\gG=(\gN,\gE_Q)$ be an undirected connected graph without self-loops, with $N\geq3$. As above, $\gE_Q$ is represented by symmetric ordered pairs. Let $\gN_i:=\{j:(i,j)\in\gE_Q\}$ be the set of neighbors of player $i$, and define the controller set by
    \[
        \gN_C:=\{i\in\gN:(i,j)\in\gE_Q
        \text{ for every }j\in\gN\setminus\{i\}\},
    \]
    which represents the set of players that connect to all other players and may be empty. 
    
    A finite-horizon Markov game has networked separable structure $\gG$ when the following conditions hold for all $h\in[H]$ and $s\in\gS$.

    (1) The rewards admit the decomposition
    \[
        r_{h,i}(s,\va)=\sum_{j\in\gN_i}r_{h,i,j}(s,a_i,a_j),
        \qquad i\in\gN,
    \]
    where $r_{h,i,j}:\gS\times\gA_i\times\gA_j\to[-1,1]$. For pairs outside $\gE_Q$, set $r_{h,i,j}\equiv0$.

    (2) The transition kernel has the form
    \[
        \sP_h(\cdot|s,\va)=
        \begin{cases}
            \displaystyle\sum_{j\in\gN_C}w_{h,j}(s)\sP_{h,j}(\cdot|s,a_j),
                & \gN_C\neq\emptyset,\\
            \sP_{h,o}(\cdot|s), & \gN_C=\emptyset.
        \end{cases}
    \]
    In the first case, $w_{h,j}(s)\geq0$ and $\sum_{j\in\gN_C}w_{h,j}(s)=1$, and each $\sP_{h,j}(\cdot|s,a_j)$ is a probability distribution on $\gS$. In the second case, $\sP_{h,o}(\cdot|s)$ is a probability distribution independent of every player's action.

    The game is zero-sum if, additionally,
    \[
        \sum_{i=1}^{N}r_{h,i}(s,\va)=0
        \qquad\text{for all }h\in[H],\ s\in\gS,\ \va\in\gA.
    \]
\end{definition}
Each player has at most $N-1$ neighbors, so the normalized edge rewards imply $|r_{h,i}(s,\va)|\leq N$. We discuss three structural properties of the games of interest here.

\paragraph{Multiple controller.}
For a nonempty controller set, this definition extends the game setting of a single controller in~\citep{kalogiannis2023zero}, to the setting of multiple controllers. The dynamics first sample a controller with probabilities $\{w_{h,\cdot}(s)\}$ and then follow the sampled player's transition kernel. Putting all weights on a single player recovers the trivial single-controller case. For an empty controller set, the transition is independent of the actions, and Eq.~(\ref{eq:edge-transition}) distributes it evenly over the player's incident edges.

\paragraph{Decomposable transitions.}
When $\gN_C\neq\emptyset$, define
\[
    \sF_{h,j}(\cdot|s,a_j):=
    \begin{cases}
        w_{h,j}(s)\sP_{h,j}(\cdot|s,a_j), & j\in\gN_C,\\
        0, & j\notin\gN_C.
    \end{cases}
\]
For every $(i,j)\in\gE_Q$, define the nonnegative edge transition measure
\begin{equation}
    \label{eq:edge-transition}
    \begin{aligned}
    &\sK_{h,i,j}(\cdot|s,a_i,a_j)\\
    &:=\begin{cases}
        \displaystyle\frac{1}{|\gN_i|}\sF_{h,i}(\cdot|s,a_i)
        +\sF_{h,j}(\cdot|s,a_j), & \gN_C\neq\emptyset,\\[4pt]
        \displaystyle\frac{1}{|\gN_i|}\sP_{h,o}(\cdot|s),
            & \gN_C=\emptyset.
    \end{cases}
    \end{aligned}
\end{equation}
Connectivity ensures $|\gN_i|\geq1$. When the controller set is nonempty, every player is adjacent to all controllers other than itself. Thus, in both cases,
\begin{equation}
    \label{eq:edge-transition-sum}
    \sum_{j\in\gN_i}\sK_{h,i,j}(\cdot|s,a_i,a_j)
    =\sP_h(\cdot|s,\va).
\end{equation}

\paragraph{Decomposable $Q$-functions.}
For any policy profile $\pi$, define
\begin{equation}
    \label{eq:edge-policy-q}
    \begin{split}
    Q_{h,i,j}^{\pi}(s,a_i,a_j)
    :=&r_{h,i,j}(s,a_i,a_j)\\
    &+\sK_{h,i,j}(\cdot|s,a_i,a_j)V_{h+1,i}^{\pi}(\cdot).
    \end{split}
\end{equation}
The transition identity above gives
\[
    Q_{h,i}^{\pi}(s,\va)
    =\sum_{j\in\gN_i}Q_{h,i,j}^{\pi}(s,a_i,a_j).
\]
Moreover, the decomposed transition bound and $\|V_{h+1,i}^{\pi}\|_\infty\leq N(H-h)$ yield the value bound
\begin{equation}
    \label{eq:edge-q-bound}
    \|Q_{h,i,j}^{\pi}\|_\infty
    \leq1+N(H-h)\leq NH.
\end{equation}

\paragraph{Remark.}
Theorem~1 of~\citep{park2023multi} shows that computing stationary approximate equilibria in the infinite-horizon discounted zero-sum networked separable Markov games is \texttt{PPAD}-hard in general. Therefore, we only focus on the analysis for the finite-horizon setting in this work.

\subsection{Nash equilibrium and quantal response equilibrium}
We introduce Nash equilibrium (NE) and quantal response equilibrium (QRE). Following the role of QRE in~\citep{cen2022faster}, we use it as an intermediate target for analyzing convergence to an approximate NE.

For a fixed policy profile $\pi$, define the best-response values of player $i$ by
\[
    V_{h,i}^{\dag,\pi_{-i}}(s):=\sup_{\pi_i}V_{h,i}^{\pi_i,\pi_{-i}}(s),
\]
\[
    Q_{h,i}^{\dag,\pi_{-i}}(s,\va):=\sup_{\pi_i}Q_{h,i}^{\pi_i,\pi_{-i}}(s,\va),
\]
where the supremum is over Markov policies of player $i$. An exact Nash equilibrium is an action profile where each player's strategy is the best response. In practice, we describe the error towards learning the Nash equilibrium using the approximate Nash equilibrium.
\begin{definition}[$\epsilon$-approximate Nash equilibrium]
    \label{NGdef}
    A Markov policy profile $\pi$ is an $\epsilon$-approximate Nash equilibrium if
    \[
        \textnormal{NE-Gap}(\pi)
        :=\max_{s\in\gS}\max_{i\in[N]}
        \left\{V_{1,i}^{\dag,\pi_{-i}}(s)-V_{1,i}^{\pi}(s)\right\}
        \leq\epsilon.
    \]
\end{definition}

For a finite normal-form game with utilities $(u_i)_{i\in[N]}$ and temperature $\tau>0$, a QRE is a product distribution $\pi^{*,\tau}$ satisfying
\begin{equation}
    \label{qre}
    \pi_i^{*,\tau}(a_i)
    =\frac{\exp(u_i(a_i,\pi_{-i}^{*,\tau})/\tau)}
    {\sum_{a\in\gA_i}\exp(u_i(a,\pi_{-i}^{*,\tau})/\tau)}.
\end{equation}
Equivalently, it is a Nash equilibrium of the entropy-regularized utilities $$u_{i,\tau}(\pi):=u_i(\pi)+\tau\gH(\pi_i),$$ where $\gH(\pi_i):=-\sum_{a\in\gA_i}\pi_i(a)\log\pi_i(a)$ is Shannon entropy~\citep{mertikopoulos2016learning}. Although an NE of a zero-sum networked separable game need not be unique~\citep{cai2016zero}, its positive-temperature QRE is unique. The following property is also contained in Theorem~4.1 of~\citep{leonardos2021exploration}, and a direct proof is given in Appendix~\ref{app:proofs-main}.
\begin{proposition}[Existence and uniqueness of QRE in zero-sum networked separable games]
    \label{prop:qre-well-defined}
    Fix any finite zero-sum networked separable normal-form game with utilities $(u_i)_{i\in[N]}$. For every $\tau>0$, the QRE defined by Eq.~(\ref{qre}) exists, has full support, and is unique.
\end{proposition}

We define the QRE at each stage of the finite-horizon zero-sum networked separable games. Let $V_{H+1,i}^{*,\tau}\equiv0$, and define the $Q$-functions at each stage $h=H$ as
\begin{equation}
    \label{RQ}
    Q_{h,i}^{*,\tau}(s,\va)
    :=r_{h,i}(s,\va)+\sP_h(\cdot|s,\va)V_{h+1,i}^{*,\tau}(\cdot).
\end{equation}
It is easily seen that $Q_h^{*,\tau}(s)$ is zero-sum, and Eq.~(\ref{eq:edge-transition-sum}) ensures its networked separable property. Considering the game solely at stage $h$ as a normal-form zero-sum networked separable game, Proposition~\ref{prop:qre-well-defined} then gives its unique QRE, denoted by $\pi_h^{*,\tau}(s)$. Subsequently, we define
\begin{equation}
    \label{RV}
    V_{h,i}^{*,\tau}(s):=Q_{h,i}^{*,\tau}(s,\pi_h^{*,\tau}(s)),
\end{equation}
which then induces the definition of $Q_{h-1,i}^{*,\tau}$ in a similar way to Eq.~(\ref{RQ}). Backward induction then defines the QRE of every stage of the zero-sum networked separable Markov games.


Now for a normal-form game $u=(u_i)_{i\in[N]}$, a product profile $\pi$, and a fixed temperature $\tau>0$, define
\[
\begin{split}
    &\textnormal{QRE-Gap}(u,\pi)\\
    &:=\max_{i\in[N]}\max_{\pi_i'\in\Delta(\gA_i)}
    \left\{u_{i,\tau}(\pi_i',\pi_{-i})-u_{i,\tau}(\pi_i,\pi_{-i})\right\}.
\end{split}
\]

We then define the stage-wise QRE-Gap of an evolving policy $\bar\pi_h^t$ at stage $h$ as
\begin{equation}
    \textnormal{QRE-Gap}_h(\bar\pi_h^t)
    :=\max_{s\in\gS}\textnormal{QRE-Gap}(Q_h^{*,\tau}(s),\bar\pi_h^t(s)),
\end{equation}
where the $Q$-functions adopt the QRE $\pi_{h^\prime}^{*,\tau}$ for stages $h^\prime>h$, and the gap only measures the value deviation of the current policy $\bar\pi_h^t$ with the regularized utility values at stage $h$. We will show that QRE-Gap is a useful and important intermediate term in analyzing the NE-Gap.

\section{Main result}
\label{sec:main}


We now present Entropy-Regularized OMWU for solving NE in the games of interest and provide its finite-time last-iterate convergence guarantee. We establish two assumptions on the information revealed to the algorithm and the players:

\begin{assumption}[Full model information]
    \label{assump:full-model-info}
    The edge-wise reward functions $\{r_{h,i,j}:h\in[H],i\in[N],j\in\gN_i\}$ and the decomposed transition component functions $\{\sK_{h,i,j}:h\in[H],i\in[N],j\in\gN_i\}$ are known to all the players in the algorithm.
\end{assumption}
\begin{assumption}[Policy revealing mechanism]
    \label{assump:policy-revealed}
    At each iteration, players are aware of the current policy profile $\bar\pi^t$. Then based on the model information and the policies of other players, they conduct policy updates, release their updated policies $\bar\pi^{t+1}$ simultaneously, and move to the next iteration.
\end{assumption}




\begin{algorithm*}[t]
    \caption{Entropy-Regularized OMWU for Zero-Sum Networked Separable Markov Games}
    \label{Algo}
    \KwIn{Learning rate $\eta>0$, regularization parameter $\tau>0$, edge rewards $(r_{h,i,j})$, edge transition components $(\sK_{h,i,j})$.}
    \KwOut{Strategy profile $\bar\pi^T$.}
    \BlankLine
    \SetAlgoNoLine
    \textbf{Initialization:} $Q_{h,i,j}^0=r_{h,i,j}$ for all $h\in[H]$, $i\in[N]$, and $j\in\gN_i$; $V_{H+1,i}^t\equiv0$ for all $i\in[N]$ and $t\geq0$; $\pi_h^0=\bar\pi_h^0$ is the uniform product distribution for all $h\in[H]$.\\
    \For{$t=0,1,\dots,T-1$}{
        Define $Q_{h,i}^t:=\sum_{j\in\gN_i}Q_{h,i,j}^t$ for all $h\in[H]$ and $i\in[N]$.\\
        For all $h\in[H]$, $i\in[N]$, $s\in\gS$, and $a_i\in\gA_i$,
        \begin{equation}
            \label{omwu2}
            \bar\pi_{h,i}^{t+1}(a_i|s)
            \propto
            \pi_{h,i}^{t}(a_i|s)^{1-\eta\tau}
            \exp\!\left(\eta Q_{h,i}^{t}(s,a_i,\bar\pi_{h,-i}^{t})\right).
        \end{equation}

        \For{$h=H,H-1,\dots,1$}{
            For all $s\in\gS$, $i\in[N]$, $j\in\gN_i$, $a_i\in\gA_i$, and $a_j\in\gA_j$,
            \begin{equation}
                T_{h,i,j}^{t+1}(s,a_i,a_j)
                =r_{h,i,j}(s,a_i,a_j)
                +\sK_{h,i,j}(\cdot|s,a_i,a_j)V_{h+1,i}^{t+1}(\cdot),
            \end{equation}
            \begin{equation}
                \label{rule1}
                Q_{h,i,j}^{t+1}(s,a_i,a_j)
                =(1-\eta\tau)Q_{h,i,j}^{t}(s,a_i,a_j)
                +\eta\tau T_{h,i,j}^{t+1}(s,a_i,a_j).
            \end{equation}

            For all $s\in\gS$ and $i\in[N]$,
            \begin{equation}
                \label{rule3}
                V_{h,i}^{t+1}(s)
                =\sum_{j\in\gN_i}
                \bar\pi_{h,i}^{t+1}(\cdot|s)^\top
                Q_{h,i,j}^{t+1}(s,\cdot,\cdot)
                \bar\pi_{h,j}^{t+1}(\cdot|s).
            \end{equation}
        }
        Define $Q_{h,i}^{t+1}:=\sum_{j\in\gN_i}Q_{h,i,j}^{t+1}$ for all $h\in[H]$ and $i\in[N]$.\\
        For all $h\in[H]$, $i\in[N]$, $s\in\gS$, and $a_i\in\gA_i$,
        \begin{equation}
            \label{omwu1}
            \pi_{h,i}^{t+1}(a_i|s)
            \propto
            \pi_{h,i}^{t}(a_i|s)^{1-\eta\tau}
            \exp\!\left(\eta Q_{h,i}^{t+1}(s,a_i,\bar\pi_{h,-i}^{t+1})\right).
        \end{equation}
    }
    \Return{$\bar\pi^T$}
\end{algorithm*}

The algorithm proceeds as follows. Each player $i$ maintains two streams of policies, $\bar\pi_{h,i}^t$ and $\pi_{h,i}^t$ at each iteration $t$, where $\bar\pi_{h,i}^t$ is the the player's formal policy and $\pi_{h,i}^t$ is auxiliary for computation. The initialization of the two policies at $t=0$ are uniform distributions. At each iteration $t$, the players first update their optimistic policies $\bar\pi_{h,i}^{t+1}$ using the current $Q$-estimates $Q_{h,i}^t$ and the previous policies $\pi_{h,i}^t$. Then they update the $Q$-estimates $Q_{h,i,j}^{t+1}$ for each edge $(i,j)$ using the new optimistic policies $\bar\pi_{h,i}^{t+1}$ and $\bar\pi_{h,j}^{t+1}$. Finally, they update their auxiliary policies $\pi_{h,i}^{t+1}$ using the new $Q$-estimates $Q_{h,i}^{t+1}$ and the new optimistic policies $\bar\pi_{h,-i}^{t+1}$. The algorithm returns the final optimistic policy profile $\bar\pi^T$ after $T$ iterations. In theory, we show that the NE-Gap decays in an exponential manner following this algorithm, resulting in an near-linear convergence rate.


\begin{theorem}[Last-iterate guarantee for ER-OMWU]
    \label{mainthm}
    Suppose that the game is a finite-horizon zero-sum networked separable Markov game with $A\geq2$. Let Algorithm~\ref{Algo} be run with $\tau>0$ and
    \[
        0<\eta\leq
        \min\left\{\frac{1}{16N^2H},\frac{1}{2\tau}\right\}.
    \]
    Then there is a constant $C_{\mathrm{poly}}=C_{\mathrm{poly}}(N,H,A,\eta,\tau)>0$, independent of $t$, such that for every $t\geq1$,
    \[
        \boxed{
        \textnormal{NE-Gap}(\bar\pi^t)
        \leq
        C_{\mathrm{poly}}(t+H+2)^{2H}(1-\eta\tau)^{t/2}
        +\tau H\log A.
        }
    \]
    In particular, for a target accuracy $\epsilon\in(0,1]$, choose
    \[
        \eta=\frac{1}{16N^2H},
        \qquad
        \tau=\frac{\epsilon}{2H\log A}.
    \]
    The constant constructed in the proof satisfies
    \[
        \log C_{\mathrm{poly}}
        =O\!\left(H\log\frac{NH\log A}{\epsilon}\right).
    \]
    Consequently, Algorithm~\ref{Algo} returns an $\epsilon$-Nash equilibrium after
    \[
        \boxed{
        t=\widetilde O\!\left(\frac{N^2H^3\log A}{\epsilon}\right)
        }
    \]
    iterations, where $\widetilde O(\cdot)$ hides universal constants and logarithmic factors in $N,H,A$, and $1/\epsilon$.
\end{theorem}

\paragraph{Proof sketch of the main theorem.}
Lemma~\ref{decomp} decomposes the upper bound of the Nash-equilibrium gap by stage-wise QRE gaps, policy divergence, and the entropy bias. Lemma~\ref{keygoal} then controls policy divergence and $Q$-function error all together through two coupled recursions (stated in Lemmas~\ref{iter1} and~\ref{iter2}). Solving these recursions backward over the horizon yields exponentially decaying errors with polynomial constants. Finally, Lemma~\ref{gregapest} derives the QRE-gap bounds from these established bounds.


\paragraph{Auxiliary notations.}
For every $h\in[H]$ and $t\geq0$, define
\[
    \bar L_h^t:=\max_{s\in\gS}
    \KL(\pi_h^{*,\tau}(s)||\bar\pi_h^t(s)),
\]
\[
    \delta_h^t:=
    \max_{s\in\gS}\max_{i\in[N]}\max_{\va\in\gA}
    \left|Q_{h,i}^t(s,\va)-Q_{h,i}^{*,\tau}(s,\va)\right|.
\]
$\bar L_h^t$ measures the policy divergence, 
and $\delta_h^t$ measures the estimate error of the aggregate $Q$-functions. We can first decompose the NE-gap into three terms.

\begin{lemma}[Decomposition of NE-Gap]
    \label{decomp}
    For every $t\geq0$,
    \[
    \begin{aligned}
        &\textnormal{NE-Gap}(\bar\pi^t)
        \leq{}
        \sum_{h=1}^H\textnormal{QRE-Gap}_h(\bar\pi_h^t)\\
        &+2NH\sqrt{2N}
        \sum_{h=1}^H\sum_{\ell=h+1}^H\sqrt{\bar L_{\ell}^t}
        +\tau H\log A.
    \end{aligned}
    \]
\end{lemma}
We then bound the first two terms with the following two lemmas, respectively.
\begin{lemma}[Tracking bounds]
    \label{keygoal}
    Under the conditions of Theorem~\ref{mainthm}, there exists a constant $C= O(N^3H\log A/\tau)$ and $C_h \leq C^{H-h}$ such that for every $h\in[H]$ and $t\geq0$,
    \[
        \delta_h^t+\sqrt{\bar L_h^t}
        \leq
        C_h(t+H-h+2)^{H-h}(1-\eta\tau)^{t/2}.
    \]
\end{lemma}

\begin{lemma}[QRE-gap estimate]
    \label{gregapest}
    Under the conditions of Theorem~\ref{mainthm}, there exists a constant $C= O(N^4H^2\log A/\tau)$ and $G_h \leq C^{2(H-h)}$ such that for every $h\in[H]$ and $t\geq 0$,
    \[
        \textnormal{QRE-Gap}_h(\bar\pi_h^t)
        \leq
        G_h(t+H-h+2)^{2(H-h)}(1-\eta\tau)^t.
    \]
    In particular, if $\eta=\Theta(1/(N^2H))$, then for every $\varepsilon\in(0,1]$,
    \(
        \max_{h\in[H]}\textnormal{QRE-Gap}_h(\bar\pi_h^t)
        \leq\varepsilon
    \)
    after
    \(
        t=\widetilde O\!\left(
            \frac{N^2H^2}{\tau}
            \log\frac{1}{\varepsilon}
        \right)
    \)
    iterations. Here $\widetilde O(\cdot)$ hides additional logarithmic factors in $N,H,A,\tau^{-1}$.
\end{lemma}

\paragraph{Proof of Theorem~\ref{mainthm}.}
Write $\beta=1-\eta\tau\in[1/2,1)$. By Lemmas~\ref{decomp}, \ref{keygoal}, and~\ref{gregapest}, for every $t\geq1$,
\[
\begin{aligned}
    &\textnormal{NE-Gap}(\bar\pi^t)
    \leq{}
    \sum_{h=1}^H G_h(t+H-h+2)^{2(H-h)}\beta^t
    \\
    &+2NH\sqrt{2N}
    \sum_{\ell=2}^H(\ell-1)C_\ell
    (t+H-\ell+2)^{H-\ell}\beta^{t/2}\\
    &+\tau H\log A.
\end{aligned}
\]
Using $\beta^t\leq\beta^{t/2}$ and bounding both polynomial factors by $(t+H+2)^{2H}$, define
\[
    C_{\mathrm{poly}}
    :=\sum_{h=1}^H G_h
    +2NH\sqrt{2N}\sum_{\ell=2}^H(\ell-1)C_\ell,
\]
where the second summation is zero when $H=1$. This gives
\[
    \textnormal{NE-Gap}(\bar\pi^t)
    \leq C_{\mathrm{poly}}(t+H+2)^{2H}\beta^{t/2}
    +\tau H\log A.
\]
For $\eta=1/(16N^2H)$ and $\tau=\epsilon/(2H\log A)$, the entropy bias is \[\tau H\log A=\frac{\epsilon}{2}.\]
According to Lemmas~\ref{keygoal} and~\ref{gregapest}, there exists a constant $C=O(N^4H^2\log A/\tau)$ such that $C_h,G_h \leq C^{2(H-h)}$. Hence,
\(
    C_{\mathrm{poly}}\leq O(N^2H^3C^{2H} ),
\)
indicating that 
\[
    \log C_{\mathrm{poly}} = \widetilde O\left(H\log\frac{NH\log A}{\epsilon}\right).
\]
Since $\beta^{t/2}=(1-\eta\tau)^{t/2}\leq\exp(-\eta\tau t/2)$, in order for $C_{\mathrm{poly}}(t+H+2)^{2H}\exp(-\eta\tau t/2)=\epsilon/2$, it suffices to choose $t$ such that
\[
\begin{split}
    &t
    = \frac{2}{\eta\tau}\left(\log 2C_{\mathrm{poly}}-\log \epsilon
        +2H\log(t+H+2)\right)
    \\
    &=\frac{64N^2H^2\log A}{\epsilon}(\log C_{\mathrm{poly}}+\log \frac{2}{\epsilon}
        +2H\log(t+H+2))
\end{split}
\]
For a sufficiently large universal constant, this holds with
\[
    t=\widetilde O\!\left(\frac{N^2H^3\log A}{\epsilon}\right).
\] \hfill $\square$

\paragraph{Discussions.}


Theorem~\ref{mainthm} establishes a last-iterate convergence rate of $\widetilde O(1/\epsilon)$ for finding an $\epsilon$-approximate Nash equilibrium, which is a near-optimal result. This matches the convergence result of the value-iteration methods studied in \citep{park2023multi}. It also demonstrates that the near-optimal convergence rate of OMWU obtained in two-player zero-sum Markov games remains preserved in zero-sum networked separable Markov games, which can be seen as a generalized multi-player setting of the former. This can also be reflected in the convergence rate through the dependence on other parameters. The result in \citep{cen2022faster} shows that in the two-player setting, the convergence rate for learning an approximate NE is
\(
    t=\widetilde O(H^3/\epsilon).
\)
Our result matches this in terms of the exponent of $H$. The difference is that, since the game is a multi-player game, a new dependence term $N^2$ is introduced.


\section{Numerical experiments}
\label{sec:experiments}

\begin{figure*}[t]
    \centering
    \includegraphics[width=\linewidth]{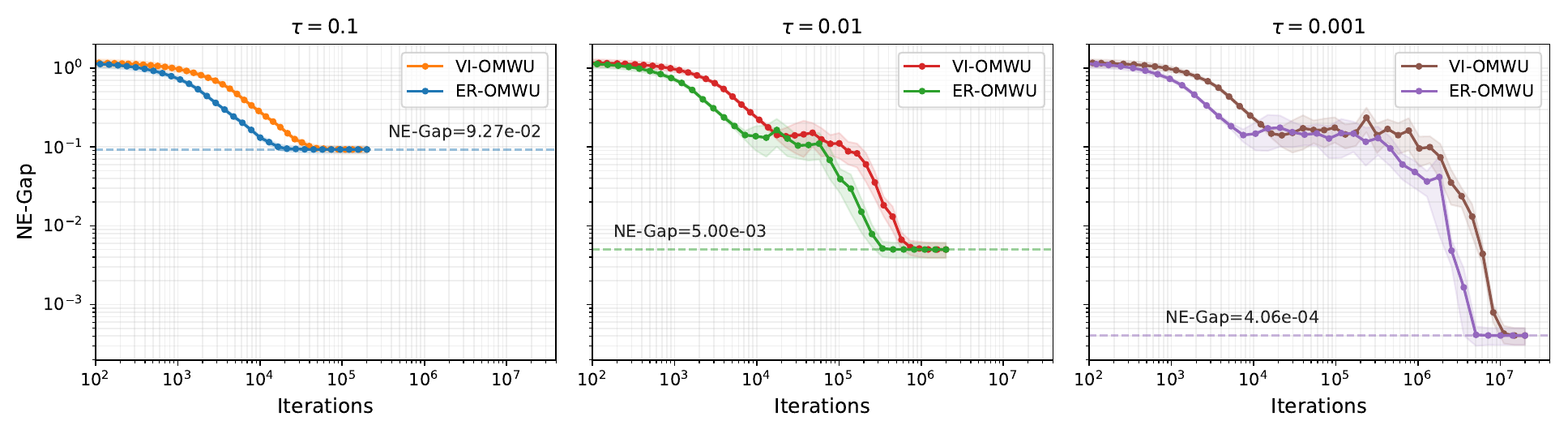}
    \caption{Convergence of the first-stage Nash-equilibrium gap for
    $\tau\in\{0.1,0.01,0.001\}$. Curves show means over 10 sampled instances on
    log-log axes, with VI-OMWU plotted using the matched update-count
    convention.}
    \label{fig:negap}
\end{figure*}

We present a numerical study illustrating two algorithmic consequences
of the theory. First, under the same entropy temperature $\tau$, ER-OMWU
converges to essentially the same NE-gap as the value-iteration-based method.
Second, as the entropy temperature decreases, the observed terminal NE-gap
decreases in the same qualitative way as the entropy-bias term in
Theorem~\ref{mainthm}. The experiments are run on zero-sum networked separable Markov games in which the interaction graph is complete, the number
of players is $N=3$, the horizon is $H=5$, the number of states is $S=2$, and
each player has $A=2$ actions at every state. We use 10 random seeds to generate
10 independent game instances. All reported statistics are averaged over these
instances, with confidence intervals shown in the figures.

We compare ER-OMWU with the value-iteration baseline VI-OMWU proposed in \citep{park2023multi}. VI-OMWU performs
backward dynamic programming and invokes an inner OMWU solver for each induced
state-stage normal-form game. Both methods are evaluated on a geometric grid with nominal size 50. The
ER-OMWU learning-rate multiplier is $0.0625$, so in this experiment
\(
    \eta_{\mathrm{ER}}
    =
    \frac{0.0625}{N^2H}
    =
    \frac{1}{720}.
\)
We evaluate the methods over the temperature grid
\(
    \tau\in\{0.5,0.1,0.05,0.01,0.005,0.001\}.
\)
For $\tau\in\{0.1,0.01,0.001\}$, ER-OMWU is run for $2\times10^5$,
$2\times10^6$, and $2\times10^7$ iterations, respectively. The complementary
temperatures $\{0.5,0.05,0.005\}$ use the same budgets. Since VI-OMWU solves
$H S=10$ state-stage games per backward pass, the plots use a matched
update-count convention: ER-OMWU is plotted against its number of single-loop
policy updates, while VI-OMWU is plotted against the aggregate number of inner
solver updates across all state-stage games. Under this convention, the
VI-OMWU inner-loop budgets for $\tau\in\{0.1,0.01,0.001\}$ are $2\times10^4$,
$2\times10^5$, and $2\times10^6$, respectively. The $y$-axis reports the exact
first-stage Markov Nash-equilibrium gap, maximized over players and initial
states.

\begin{figure}[tbp]
    \centering
    \includegraphics[width=0.8\linewidth]{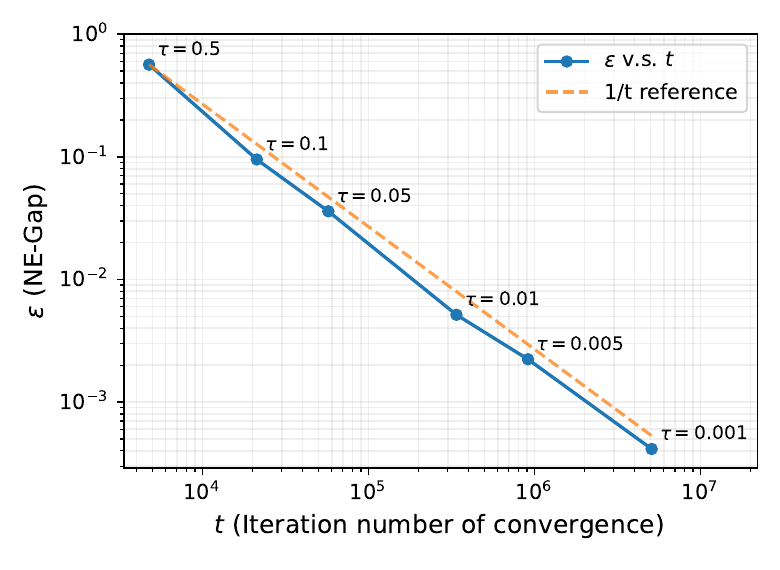}
    \caption{NE-Gap vs. Convergence iteration for ER-OMWU. The reference line has slope
    $1/T$.}
    \label{fig:linear}
\end{figure}

\Cref{fig:negap} shows the evolution of the NE-gap for the three representative
temperatures. Across all panels, ER-OMWU decreases the last-iterate NE-gap and
reaches essentially the same terminal accuracy as VI-OMWU under the matched
update-count convention. This supports the main algorithmic point: on these
instances, replacing repeated stage-game equilibrium solves with edge-wise
value tracking and a single policy trajectory preserves the relevant accuracy
behavior. The effect of the entropy temperature is also visible. Smaller
$\tau$ leads to a lower final NE-gap floor, consistent with the entropy-bias
term in Theorem~\ref{mainthm}, while requiring more updates before the curve
stabilizes. The results for complementary temperatures are reported in \Cref{fig:negap-extra-tau} in
Appendix~\ref{app:supplementary-experiments}.

\Cref{fig:linear} examines the near-linear accuracy dependence more directly.
For each $\tau$, we define the convergence iteration as the first evaluation
point at which the mean NE-gap reaches its terminal plateau, estimated from the
tail of the curve with a $5\%$ relative tolerance. The figure plots this
convergence iteration against the corresponding terminal NE-gap over all six
temperatures. After choosing
$\tau$ at the desired accuracy scale, the single-loop ER-OMWU dynamics displays
the near-linear accuracy trend predicted by Theorem~\ref{mainthm}, while
avoiding the stage-game equilibrium-solver calls used by the VI baseline.

\section{Conclusion}

We studied last-iterate equilibrium learning in finite-horizon zero-sum networked separable Markov games and proposed ER-OMWU, a single-loop method that combines symmetric optimistic policy updates with edge-wise Bellman tracking. In the known-model setting, the method finds an $\epsilon$-approximate Nash equilibrium in $\widetilde O(N^2H^3\log A/\epsilon)$ iterations without solving equilibrium problems at individual state-stage pairs. This result establishes that direct policy optimization can retain near-linear accuracy dependence in structured multi-player competition, extending the last-iterate perspective beyond two-player zero-sum games. A natural next step is to develop analogous guarantees with unknown transitions and sampled feedback, where policy learning and value estimation must be controlled jointly. More broadly, it remains to understand how much communication and structural information fast last-iterate learning requires, and which richer interaction and transition models admit similar guarantees.


\bibliographystyle{plainnat}
\bibliography{references}  

@inproceedings{NEURIPS2025_42fe68e7,
 author = {Cai, Yang and Luo, Haipeng and Wei, Chen-Yu and Zheng, Weiqiang},
 booktitle = {Advances in Neural Information Processing Systems},
 pages = {46937--46967},
 title = {From Average-Iterate to Last-Iterate Convergence in Games: A Reduction and Its Applications},
 volume = {38, Main Conference},
 year = {2025}
}

@inproceedings{cai2011minmax,
  title={On minmax theorems for multiplayer games},
  author={Cai, Yang and Daskalakis, Constantinos},
  booktitle={Proceedings of the twenty-second annual ACM-SIAM symposium on Discrete algorithms},
  pages={217--234},
  year={2011},
  organization={SIAM}
}

@article{leonardos2021exploration,
  title={Exploration-exploitation in multi-agent competition: convergence with bounded rationality},
  author={Leonardos, Stefanos and Piliouras, Georgios and Spendlove, Kelly},
  journal={Advances in Neural Information Processing Systems},
  volume={34},
  pages={26318--26331},
  year={2021}
}

@article{cen2022faster,
  title={Faster Last-iterate Convergence of Policy Optimization in Zero-Sum Markov Games},
  author={Shicong Cen and Yuejie Chi and Simon Shaolei Du and Lin Xiao},
  journal={The Eleventh International Conference on Learning Representations },
  year={2023}
}

@inproceedings{erez2023regret,
  title={Regret minimization and convergence to equilibria in general-sum markov games},
  author={Erez, Liad and Lancewicki, Tal and Sherman, Uri and Koren, Tomer and Mansour, Yishay},
  booktitle={International Conference on Machine Learning},
  pages={9343--9373},
  year={2023},
  organization={PMLR}
}

@article{zhang2022policy,
  title={Policy optimization for markov games: Unified framework and faster convergence},
  author={Zhang, Runyu and Liu, Qinghua and Wang, Huan and Xiong, Caiming and Li, Na and Bai, Yu},
  journal={Advances in Neural Information Processing Systems},
  volume={35},
  pages={21886--21899},
  year={2022}
}

@inproceedings{wei2021last,
  title={Last-iterate convergence of decentralized optimistic gradient descent/ascent in infinite-horizon competitive Markov games},
  author={Wei, Chen-Yu and Lee, Chung-Wei and Zhang, Mengxiao and Luo, Haipeng},
  booktitle={Conference on learning theory},
  pages={4259--4299},
  year={2021},
  organization={PMLR}
}

@article{anagnostides2022uncoupled,
  title={Uncoupled Learning Dynamics with  O (log T)  Swap Regret in Multiplayer Games},
  author={Anagnostides, Ioannis and Farina, Gabriele and Kroer, Christian and Lee, Chung-Wei and Luo, Haipeng and Sandholm, Tuomas},
  journal={Advances in Neural Information Processing Systems},
  volume={35},
  pages={3292--3304},
  year={2022}
}

@inproceedings{anagnostides2022near,
  title={Near-optimal no-regret learning for correlated equilibria in multi-player general-sum games},
  author={Anagnostides, Ioannis and Daskalakis, Constantinos and Farina, Gabriele and Fishelson, Maxwell and Golowich, Noah and Sandholm, Tuomas},
  booktitle={Proceedings of the 54th Annual ACM SIGACT Symposium on Theory of Computing},
  pages={736--749},
  year={2022}
}

@article{daskalakis2021near,
  title={Near-optimal no-regret learning in general games},
  author={Daskalakis, Constantinos and Fishelson, Maxwell and Golowich, Noah},
  journal={Advances in Neural Information Processing Systems},
  volume={34},
  pages={27604--27616},
  year={2021}
}

@article{chen2020hedging,
  title={Hedging in games: Faster convergence of external and swap regrets},
  author={Chen, Xi and Peng, Binghui},
  journal={Advances in Neural Information Processing Systems},
  volume={33},
  pages={18990--18999},
  year={2020}
}

@article{syrgkanis2015fast,
  title={Fast convergence of regularized learning in games},
  author={Syrgkanis, Vasilis and Agarwal, Alekh and Luo, Haipeng and Schapire, Robert E},
  journal={Advances in Neural Information Processing Systems},
  volume={28},
  year={2015}
}

@article{rakhlin2013optimization,
  title={Optimization, learning, and games with predictable sequences},
  author={Rakhlin, Sasha and Sridharan, Karthik},
  journal={Advances in Neural Information Processing Systems},
  volume={26},
  year={2013}
}

@inproceedings{daskalakis2011near,
  title={Near-optimal no-regret algorithms for zero-sum games},
  author={Daskalakis, Constantinos and Deckelbaum, Alan and Kim, Anthony},
  booktitle={Proceedings of the twenty-second annual ACM-SIAM symposium on Discrete Algorithms},
  pages={235--254},
  year={2011},
  organization={SIAM}
}

@article{rubinstein2017settling,
  title={Settling the complexity of computing approximate two-player Nash equilibria},
  author={Rubinstein, Aviad},
  journal={ACM SIGecom Exchanges},
  volume={15},
  number={2},
  pages={45--49},
  year={2017},
  publisher={ACM New York, NY, USA}
}

@article{etessami2010complexity,
  title={On the complexity of Nash equilibria and other fixed points},
  author={Etessami, Kousha and Yannakakis, Mihalis},
  journal={SIAM Journal on Computing},
  volume={39},
  number={6},
  pages={2531--2597},
  year={2010},
  publisher={SIAM}
}

@article{chen2009settling,
  title={Settling the complexity of computing two-player Nash equilibria},
  author={Chen, Xi and Deng, Xiaotie and Teng, Shang-Hua},
  journal={Journal of the ACM (JACM)},
  volume={56},
  number={3},
  pages={1--57},
  year={2009},
  publisher={ACM New York, NY, USA}
}

@article{daskalakis2009complexity,
  title={The complexity of computing a Nash equilibrium},
  author={Daskalakis, Constantinos and Goldberg, Paul W and Papadimitriou, Christos H},
  journal={Communications of the ACM},
  volume={52},
  number={2},
  pages={89--97},
  year={2009},
  publisher={ACM New York, NY, USA}
}

@article{perolat2022mastering,
  title={Mastering the game of Stratego with model-free multiagent reinforcement learning},
  author={Perolat, Julien and De Vylder, Bart and Hennes, Daniel and Tarassov, Eugene and Strub, Florian and de Boer, Vincent and Muller, Paul and Connor, Jerome T and Burch, Neil and Anthony, Thomas and others},
  journal={Science},
  volume={378},
  number={6623},
  pages={990--996},
  year={2022},
  publisher={American Association for the Advancement of Science}
}

@article{brown2020combining,
  title={Combining deep reinforcement learning and search for imperfect-information games},
  author={Brown, Noam and Bakhtin, Anton and Lerer, Adam and Gong, Qucheng},
  journal={Advances in Neural Information Processing Systems},
  volume={33},
  pages={17057--17069},
  year={2020}
}

@article{brown2019superhuman,
  title={Superhuman AI for multiplayer poker},
  author={Brown, Noam and Sandholm, Tuomas},
  journal={Science},
  volume={365},
  number={6456},
  pages={885--890},
  year={2019},
  publisher={American Association for the Advancement of Science}
}

@article{brown2018superhuman,
  title={Superhuman AI for heads-up no-limit poker: Libratus beats top professionals},
  author={Brown, Noam and Sandholm, Tuomas},
  journal={Science},
  volume={359},
  number={6374},
  pages={418--424},
  year={2018},
  publisher={American Association for the Advancement of Science}
}

@article{moravvcik2017deepstack,
  title={Deepstack: Expert-level artificial intelligence in heads-up no-limit poker},
  author={Morav{\v{c}}{\'\i}k, Matej and Schmid, Martin and Burch, Neil and Lis{\`y}, Viliam and Morrill, Dustin and Bard, Nolan and Davis, Trevor and Waugh, Kevin and Johanson, Michael and Bowling, Michael},
  journal={Science},
  volume={356},
  number={6337},
  pages={508--513},
  year={2017},
  publisher={American Association for the Advancement of Science}
}

@article{bowling2015heads,
  title={Heads-up limit hold’em poker is solved},
  author={Bowling, Michael and Burch, Neil and Johanson, Michael and Tammelin, Oskari},
  journal={Science},
  volume={347},
  number={6218},
  pages={145--149},
  year={2015},
  publisher={American Association for the Advancement of Science}
}

@article{silver2017mastering,
  title={Mastering the game of go without human knowledge},
  author={Silver, David and Schrittwieser, Julian and Simonyan, Karen and Antonoglou, Ioannis and Huang, Aja and Guez, Arthur and Hubert, Thomas and Baker, Lucas and Lai, Matthew and Bolton, Adrian and others},
  journal={nature},
  volume={550},
  number={7676},
  pages={354--359},
  year={2017},
  publisher={Nature Publishing Group}
}

@article{mertikopoulos2016learning,
  title={Learning in games via reinforcement and regularization},
  author={Mertikopoulos, Panayotis and Sandholm, William H},
  journal={Mathematics of Operations Research},
  volume={41},
  number={4},
  pages={1297--1324},
  year={2016},
  publisher={INFORMS}
}

@article{park2023multi,
  title={Multi-Player Zero-Sum Markov Games with Networked Separable Interactions},
  author={Chanwoo Park and Kaiqing Zhang and Asuman E. Ozdaglar},
  journal={Thirty-seventh Conference on Neural Information Processing Systems},
  year={2023}
}

@article{kalogiannis2023zero,
  title={Zero-sum Polymatrix Markov Games: Equilibrium Collapse and Efficient Computation of Nash Equilibria},
  author={Kalogiannis, Fivos and Panageas, Ioannis},
  journal={arXiv preprint arXiv:2305.14329},
  year={2023}
}

@article{ao2022asynchronous,
  title={Asynchronous gradient play in zero-sum multi-agent games},
  author={Ao, Ruicheng and Cen, Shicong and Chi, Yuejie},
  journal={arXiv preprint arXiv:2211.08980},
  year={2022}
}

@article{cai2016zero,
  title={Zero-sum polymatrix games: A generalization of minmax},
  author={Cai, Yang and Candogan, Ozan and Daskalakis, Constantinos and Papadimitriou, Christos},
  journal={Mathematics of Operations Research},
  volume={41},
  number={2},
  pages={648--655},
  year={2016},
  publisher={INFORMS}
}

@inproceedings{yang2023ot,
title={\$O(T{\textasciicircum}\{-1\})\$ Convergence of Optimistic-Follow-the-Regularized-Leader in Two-Player Zero-Sum Markov Games },
author={Yuepeng Yang and Cong Ma},
booktitle={The Eleventh International Conference on Learning Representations },
year={2023},
url={https://openreview.net/forum?id=VWqiPBB_EM}
}

@inproceedings{anagnostides2024optimistic,
  title={Optimistic policy gradient in multi-player markov games with a single controller: Convergence beyond the minty property},
  author={Anagnostides, Ioannis and Panageas, Ioannis and Farina, Gabriele and Sandholm, Tuomas},
  booktitle={Proceedings of the AAAI Conference on Artificial Intelligence},
  volume={38},
  note={Issue 9},
  pages={9451--9459},
  year={2024}
}

@inproceedings{cai2024near,
  title={Near-optimal policy optimization for correlated equilibrium in general-sum markov games},
  author={Cai, Yang and Luo, Haipeng and Wei, Chen-Yu and Zheng, Weiqiang},
  booktitle={International Conference on Artificial Intelligence and Statistics},
  pages={3889--3897},
  year={2024},
  organization={PMLR}
}

@inproceedings{mao2024widetilde,
  title={{{$\widetilde O(T^{-1})$}} Convergence to (coarse) correlated equilibria in full-information general-sum Markov games},
  author={Mao, Weichao and Qiu, Haoran and Wang, Chen and Franke, Hubertus and Kalbarczyk, Zbigniew and Ba{\c{s}}ar, Tamer},
  booktitle={6th Annual Learning for Dynamics \& Control Conference},
  pages={361--374},
  year={2024},
  organization={PMLR}
}

@article{cai2024fast,
  title={Fast last-iterate convergence of learning in games requires forgetful algorithms},
  author={Cai, Yang and Farina, Gabriele and Grand-Cl{\'e}ment, Julien and Kroer, Christian and Lee, Chung-Wei and Luo, Haipeng and Zheng, Weiqiang},
  journal={Advances in Neural Information Processing Systems},
  volume={37},
  pages={23406--23434},
  year={2024}
}

@article{chen2024last,
  title={Last-iterate convergence of payoff-based independent learning in zero-sum stochastic games},
  author={Chen, Zaiwei and Zhang, Kaiqing and Mazumdar, Eric and Ozdaglar, Asuman and Wierman, Adam},
  journal={arXiv preprint arXiv:2409.01447},
  year={2024}
}

@article{kalogiannis2024learning,
  title={Learning equilibria in adversarial team markov games: A nonconvex-hidden-concave min-max optimization problem},
  author={Kalogiannis, Fivos and Yan, Jingming and Panageas, Ioannis},
  journal={Advances in Neural Information Processing Systems},
  volume={37},
  pages={92832--92890},
  year={2024}
}

@article{gemp2024convex,
  title={Convex markov games: A new frontier for multi-agent reinforcement learning},
  author={Gemp, Ian and Haupt, Andreas and Marris, Luke and Liu, Siqi and Piliouras, Georgios},
  journal={arXiv preprint arXiv:2410.16600},
  year={2024}
}

@article{kalogiannis2025solving,
  title={Solving zero-sum convex markov games},
  author={Kalogiannis, Fivos and Vlatakis-Gkaragkounis, Emmanouil-Vasileios and Gemp, Ian and Piliouras, Georgios},
  journal={arXiv preprint arXiv:2506.16120},
  year={2025}
}

\clearpage
\appendix
\onecolumn
\thispagestyle{empty}
\aistatstitle{Last-Iterate Convergence of Policy Dynamics in Zero-Sum Networked Separable Markov Games\\
Supplementary Materials}

\tableofcontents

\newpage

\section{Proofs for the preliminaries}
\label{app:proofs-main}

\subsection{Proof of Proposition~\ref{prop:qre-well-defined}}
\begin{proof}
    Let $\Pi:=\prod_{i=1}^{N}\Delta(\gA_i)$ and define the payoff vector $v_i(p):=(u_i(a_i,p_{-i}))_{a_i\in\gA_i}$ for $p\in\Pi$. The logit-response map $B:\Pi\to\Pi$ has coordinates
    \[
        B_i(p)(a_i):=
        \frac{\exp(v_i(p)(a_i)/\tau)}
        {\sum_{a\in\gA_i}\exp(v_i(p)(a)/\tau)}.
    \]
    The utilities are finite and continuous in mixed profiles, and $\tau>0$. Thus $B$ is continuous and every coordinate of $B(p)$ is strictly positive. The product simplex $\Pi$ is nonempty, compact, and convex, so Brouwer's fixed-point theorem gives $p=B(p)$. This is a full-support QRE.

    To prove uniqueness, let $p,q\in\Pi$ be two QREs. Their logit equations imply
    \[
        \tau(\log p_i-\log q_i)
        =v_i(p)-v_i(q)+c_i\mathbf1
    \]
    for scalars $c_i$, where logarithms are taken coordinatewise. Taking the inner product with $p_i-q_i$ removes $c_i\mathbf1$, because both $p_i$ and $q_i$ have unit mass. Summing over players yields
    \[
        \tau\sum_{i=1}^{N}
        \left[\KL(p_i\|q_i)+\KL(q_i\|p_i)\right]
        =\sum_{i=1}^{N}
        \langle p_i-q_i,v_i(p)-v_i(q)\rangle.
    \]
    All KL divergences are finite, since both QREs have full support. The right-hand side equals
    \[
    \begin{aligned}
        &\sum_i u_i(p)+\sum_i u_i(q)\\
        &\qquad-\sum_i\left[u_i(p_i,q_{-i})+u_i(q_i,p_{-i})\right]
        =0.
    \end{aligned}
    \]
    The first two sums vanish by global zero-sum structure, and the cross sum vanishes by the zero-sum networked separable identity in Lemma~\ref{sumzero2}. Nonnegativity of KL divergence and $\tau>0$ now imply $p_i=q_i$ for every $i$. This proves uniqueness.
\end{proof}

\subsection{Well-definedness and zero-sum structure of the QRE-induced returns}

\begin{lemma}
    \label{lem:qre-induced-structure}
    For every $\tau>0$, the backward construction in Eqs.~(\ref{RQ})--(\ref{RV}) is well defined and determines a unique policy profile $\pi^{*,\tau}$. For every $h\in[H]$ and $i\in[N]$,
    \[
        Q_{h,i}^{*,\tau}=Q_{h,i}^{\pi^{*,\tau}},
        \qquad
        V_{h,i}^{*,\tau}=V_{h,i}^{\pi^{*,\tau}}.
    \]
    Moreover, for every $h\in[H]$, $s\in\gS$, and $\va\in\gA$,
    \[
        \sum_{i=1}^{N}Q_{h,i}^{*,\tau}(s,\va)=0,
        \qquad
        \sum_{i=1}^{N}V_{h,i}^{*,\tau}(s)=0.
    \]
    Each stage game $Q_h^{*,\tau}(s)$ is zero-sum and networked separable. These conclusions hold for both empty and nonempty controller sets.
\end{lemma}

\begin{proof}
    We first establish well-definedness and zero-sum structure simultaneously by backward induction, before assuming a policy at the current stage. At stage $H+1$, the uniquely specified terminal values $V_{H+1,i}^{*,\tau}\equiv0$ have zero sum. Suppose that the policies at stages $h+1,\ldots,H$ and the values $V_{h+1,i}^{*,\tau}$ have already been uniquely constructed, with $\sum_iV_{h+1,i}^{*,\tau}\equiv0$.

    Equation~(\ref{RQ}) uniquely defines $Q_{h,i}^{*,\tau}$. For every $(s,\va)$,
    \[
    \begin{aligned}
        \sum_{i=1}^{N}Q_{h,i}^{*,\tau}(s,\va)
        &=\sum_{i=1}^{N}r_{h,i}(s,\va)
        +\sP_h(\cdot|s,\va)
        \sum_{i=1}^{N}V_{h+1,i}^{*,\tau}(\cdot)\\
        &=0.
    \end{aligned}
    \]
    Independently of any choice of the stage-$h$ policy, define the provisional edge utilities
    \[
        \widetilde Q_{h,i,j}(s,a_i,a_j)
        :=r_{h,i,j}(s,a_i,a_j)
        +\sK_{h,i,j}(\cdot|s,a_i,a_j)V_{h+1,i}^{*,\tau}(\cdot).
    \]
    Equation~(\ref{eq:edge-transition-sum}) gives
    \[
        Q_{h,i}^{*,\tau}(s,\va)
        =\sum_{j\in\gN_i}\widetilde Q_{h,i,j}(s,a_i,a_j).
    \]
    Hence $Q_h^{*,\tau}(s)$ is a finite zero-sum networked separable normal-form game. Proposition~\ref{prop:qre-well-defined} now applies and supplies its unique full-support QRE $\pi_h^{*,\tau}(s)$. Equation~(\ref{RV}) uniquely defines $V_{h,i}^{*,\tau}(s)$, and taking expectation under the common product distribution $\pi_h^{*,\tau}(s)$ gives
    \[
        \sum_iV_{h,i}^{*,\tau}(s)
        =\sum_iQ_{h,i}^{*,\tau}(s,\pi_h^{*,\tau}(s))=0.
    \]
    This closes the induction and constructs the entire policy $\pi^{*,\tau}$ without circularly assuming its existence.

    With this policy now fixed, a second backward induction identifies the constructed values with its ordinary returns. At $H+1$, both terminal values are zero. If $V_{h+1,i}^{*,\tau}=V_{h+1,i}^{\pi^{*,\tau}}$, then
    \[
    \begin{aligned}
        Q_{h,i}^{*,\tau}(s,\va)
        &=r_{h,i}(s,\va)+\sP_h(\cdot|s,\va)
        V_{h+1,i}^{\pi^{*,\tau}}(\cdot)\\
        &=Q_{h,i}^{\pi^{*,\tau}}(s,\va).
    \end{aligned}
    \]
    Taking expectation under $\pi_h^{*,\tau}(s)$ proves $V_{h,i}^{*,\tau}=V_{h,i}^{\pi^{*,\tau}}$. The same equality of continuation values shows that each provisional edge utility is $Q_{h,i,j}^{\pi^{*,\tau}}$. This establishes all claimed identities.
\end{proof}

\section{Proof of Lemma~\ref{decomp}}
\begin{proof}
    Fix $t\geq0$. For every $h\in[H]$, $i\in[N]$, and $s\in\gS$, define the one-step deviation advantage
    \[
        a_{h,i}^t(s)
        :=
        \max_{\nu_i\in\Delta(\gA_i)}
        \left\langle
            \nu_i-\bar\pi_{h,i}^t(\cdot|s),
            Q_{h,i}^{\bar\pi^t}
            (s,\cdot,\bar\pi_{h,-i}^t(s))
        \right\rangle.
    \]
    This quantity is nonnegative because $\nu_i=\bar\pi_{h,i}^t(\cdot|s)$ is feasible.

    Fix a player $i$, an initial state $s_1$, and a Markov deviation policy $\mu_i$. Applying the Bellman equations to
    \[
        W_{h,i}(s)
        :=V_{h,i}^{\mu_i,\bar\pi_{-i}^t}(s)
        -V_{h,i}^{\bar\pi^t}(s)
    \]
    gives
    \[
    \begin{aligned}
        W_{h,i}(s)
        &=
        \left\langle
            \mu_{h,i}(\cdot|s)-\bar\pi_{h,i}^t(\cdot|s),
            Q_{h,i}^{\bar\pi^t}
            (s,\cdot,\bar\pi_{h,-i}^t(s))
        \right\rangle
        \\
        &\quad+
        \sP_h(\cdot|s,\mu_{h,i}(s),\bar\pi_{h,-i}^t(s))
        W_{h+1,i}(\cdot).
    \end{aligned}
    \]
    Since $W_{H+1,i}\equiv0$, iterating this identity yields
    \[
    \begin{aligned}
        &V_{1,i}^{\mu_i,\bar\pi_{-i}^t}(s_1)
        -V_{1,i}^{\bar\pi^t}(s_1)
        \\
        &=
        \sum_{h=1}^{H}
        \mathbb E_{\mu_i,\bar\pi_{-i}^t}
        \left[
            \left\langle
                \mu_{h,i}(\cdot|s_h)-\bar\pi_{h,i}^t(\cdot|s_h),
                Q_{h,i}^{\bar\pi^t}
                (s_h,\cdot,\bar\pi_{h,-i}^t(s_h))
            \right\rangle
            \,\middle|\,s_1
        \right]
        \\
        &\leq
        \sum_{h=1}^{H}\max_{s\in\gS}a_{h,i}^t(s).
    \end{aligned}
    \]
    The fixed policies of the other players induce a finite-horizon MDP for player $i$, so a Markov best response attains the supremum over deviation policies. Taking that supremum and then the maxima over $i$ and $s_1$ gives
    \begin{equation}
        \label{eq:ne-gap-stage-max}
        \textnormal{NE-Gap}(\bar\pi^t)
        \leq
        \max_{i\in[N]}\sum_{h=1}^{H}\max_{s\in\gS}a_{h,i}^t(s)
        \leq
        \sum_{h=1}^{H}\max_{i\in[N]}\max_{s\in\gS}a_{h,i}^t(s).
    \end{equation}

    For each $h,i$, set
    \[
        E_{h,i}^t
        :=\left\|Q_{h,i}^{\bar\pi^t}-Q_{h,i}^{*,\tau}\right\|_\infty.
    \]
    For any $s\in\gS$ and $\nu_i\in\Delta(\gA_i)$, the bound
    $\|\nu_i-\bar\pi_{h,i}^t(\cdot|s)\|_1\leq2$ implies
    \[
    \begin{aligned}
        &\left\langle
            \nu_i-\bar\pi_{h,i}^t(\cdot|s),
            Q_{h,i}^{\bar\pi^t}(s,\cdot,\bar\pi_{h,-i}^t(s))
        \right\rangle
        \\
        &\leq
        \left\langle
            \nu_i-\bar\pi_{h,i}^t(\cdot|s),
            Q_{h,i}^{*,\tau}(s,\cdot,\bar\pi_{h,-i}^t(s))
        \right\rangle
        +2E_{h,i}^t
        \\
        &\leq
        \textnormal{QRE-Gap}_h(\bar\pi_h^t)
        +\tau\gH(\bar\pi_{h,i}^t(\cdot|s))
        -\tau\gH(\nu_i)
        +2E_{h,i}^t
        \\
        &\leq
        \textnormal{QRE-Gap}_h(\bar\pi_h^t)
        +\tau\log A+2E_{h,i}^t.
    \end{aligned}
    \]
    Here the second inequality is the definition of the stage-wise QRE gap after adding and subtracting the entropy terms, and the last inequality uses
    $0\leq\gH(\nu_i)$ and $\gH(\bar\pi_{h,i}^t(\cdot|s))\leq\log A$.
    Maximizing over the deviation and substituting into \eqref{eq:ne-gap-stage-max} yields
    \[
        \textnormal{NE-Gap}(\bar\pi^t)
        \leq
        \sum_{h=1}^{H}\textnormal{QRE-Gap}_h(\bar\pi_h^t)
        +2\sum_{h=1}^{H}\max_{i\in[N]}E_{h,i}^t
        +\tau H\log A.
    \]

    By Lemma~\ref{lem:qre-induced-structure},
    $Q^{*,\tau}=Q^{\pi^{*,\tau}}$. Lemma~\ref{lem:q-policy-perturb}, applied to $\mu=\bar\pi^t$ and $\pi=\pi^{*,\tau}$, therefore gives uniformly in $i$,
    \[
        E_{h,i}^t
        \leq NH\sqrt{2N}
        \sum_{\ell=h+1}^{H}\bar X_\ell^t.
    \]
    Combining the last two displays proves
    \[
    \begin{aligned}
        \textnormal{NE-Gap}(\bar\pi^t)
        \leq{}&
        \sum_{h=1}^{H}\textnormal{QRE-Gap}_h(\bar\pi_h^t)
        \\
        &+2NH\sqrt{2N}
        \sum_{h=1}^{H}\sum_{\ell=h+1}^{H}\bar X_\ell^t
        +\tau H\log A.
    \end{aligned}
    \]
    The double sum is empty when $H=1$.
\end{proof}

\section{Proof of the tracking bounds}
\label{app:tracking-bounds}

The purpose of this section is to prove Lemma~\ref{keygoal}. The argument has two stages. First, we establish two elementary recursions. Lemma~\ref{iter1} shows that the KL-type policy-divergence variables contract under the entropy-regularized OMWU update, up to an additive term depending on the current $Q$-function error. Lemma~\ref{iter2} shows that the $Q$-function error at stage $h$ is controlled by the next-stage $Q$-function error and the next-stage policy perturbation. Second, we solve these two coupled recursions by backward induction over the horizon, yielding the exponential bounds stated in Lemma~\ref{keygoal}.

We first define the key variables to estimate:
\[
    \bar L_h^t:=\max_{s\in\gS}
    \KL(\pi_h^{*,\tau}(s)||\bar\pi_h^t(s)),
\]
\[
    L_h^t:=\max_{s\in\gS}
    \left\{
        \KL(\pi_h^{*,\tau}(s)||\pi_h^t(s))
        +(1-4\eta N^2H)\KL(\pi_h^t(s)||\bar\pi_h^t(s))
    \right\},
\]
\[
    \bar X_h^t:=\sqrt{\bar L_h^t},\qquad X_h^t:=\sqrt{L_h^t},
\]
\[
    \delta_h^t:=
    \max_{s\in\gS}\max_{i\in[N]}\max_{\va\in\gA}
    \left|Q_{h,i}^t(s,\va)-Q_{h,i}^{*,\tau}(s,\va)\right|.
\]

Then we introduce some auxiliary notations.
\[
    \beta:=1-\eta\tau,
    \qquad
    \alpha:=\sqrt{\beta},
\]
\[
    \Gamma_\eta:=
    256N^2\eta\left(1+1/\tau\right),
    \qquad
    \omega_\eta:=\sqrt{\Gamma_\eta}.
\]

\begin{lemma}[Policy-divergence recursion]
    \label{iter1}
    Under the conditions of Theorem~\ref{mainthm}, for every $h\in[H]$ and $t\geq0$,
    \[
        L_h^{t+1}
        \leq
        \beta L_h^t
        +
        \Gamma_\eta
        \left(
            \delta_h^t+2\delta_h^{t+1}
        \right)^2,
        \qquad
        \bar L_h^{t+1}
        \leq
        4L_h^t
        +
        \Gamma_\eta
        \left(
            \delta_h^t+2\delta_h^{t+1}
        \right)^2.
    \]
    Moreover, for every $s\in\gS$,
    \[
        \frac{\eta\tau}{2}
        \KL(\bar\pi_h^{t+1}(s)||\pi_h^{*,\tau}(s))
        \leq
        \beta L_h^t
        +
        \Gamma_\eta
        \left(
            \delta_h^t+2\delta_h^{t+1}
        \right)^2.
    \]
    Consequently,
    \[
        X_h^{t+1}
        \leq
        \alpha X_h^t
        +
        \omega_\eta
        \left(
            \delta_h^t+2\delta_h^{t+1}
        \right),
        \qquad
        \bar X_h^{t+1}
        \leq
        2X_h^t
        +
        \omega_\eta
        \left(
            \delta_h^t+2\delta_h^{t+1}
        \right).
    \]
\end{lemma}

\begin{lemma}[$Q$-tracking recursion]
    \label{iter2}
    Under the conditions of Theorem~\ref{mainthm}, for all $h\in[H]$ and $t\geq0$,
    \[
        \delta_h^t
        \leq
        \eta\tau
        \sum_{k=0}^{t-1}
        \beta^k
        \left(
            \delta_{h+1}^{t-k}
            +
            NH\sqrt{2N}\,\bar X_{h+1}^{t-k}
        \right)
        +
        NH\beta^t,
    \]
    where $\delta_{H+1}^t=0$ and $\bar X_{H+1}^t=0$ by convention.
\end{lemma}

In the rest of this section, we will first show how to obtain Lemma~\ref{keygoal} using these two lemmas, and then provide the proofs of these two lemmas.

\subsection{Proof of Lemma~\ref{keygoal}}

We state and prove a stronger version of Lemma~\ref{keygoal} here.
\begin{lemma}[Tracking bounds (stronger version)]
    \label{keygoal2}
    Under the conditions of Theorem~\ref{mainthm}, there exists a constant $C= O(N^3H\log A/\tau)$ and $C_h \leq C^{H-h}$ such that for every $h\in[H]$ and $t\geq0$,
    \[
        \delta_h^t+\sqrt{L_h^t}+
        \sqrt{\bar L_h^t}
        \leq
        C_h(t+H-h+2)^{H-h}(1-\eta\tau)^{t/2}.
    \]
\end{lemma}

\begin{proof}
    We introduce the auxiliary notation.
    \[
    q_h:=H-h.
    \]
    We prove the bound by backward induction on $h$.

    \paragraph{Case of $h=H$.}
    For $h=H$, because value functions of stage $H+1$ is zero, one has that for every $t\geq 0$,
    \[
        Q_{H,i}^t=Q_{H,i}^{*,\tau}=r_{H,i}.
    \]
    Hence for all $t\geq 0$, we have
    \[
        \delta_H^t=0.
    \]

    By Lemma~\ref{iter1},
    \[
        X_H^{t+1}
        \leq
        \alpha X_H^t+
        \omega_\eta \left(\delta_H^t+2\delta_H^{t+1}\right)=\alpha X_H^t.
    \]
    Let \[K_0:=\max_{h\in[H]}\max\{L_h^0,\bar L_h^0\}\leq O(N\log A),\] we have $X_H^0\leq\sqrt{K_0}$. Using this initial bound and the recursive relation \(
        X_H^{t+1}
        \leq \alpha X_H^t
    \), we obtain for every $t\geq0$,
    \[
        X_H^t\leq X_H^0 \alpha^t \leq \sqrt{K_0}\alpha^t.
    \]
    Again by Lemma~\ref{iter1},
    \[
        \bar X_H^{t}
        \leq
        2X_H^{t-1}+
        \omega_\eta\left(\delta_H^{t-1}+2\delta_H^{t}\right)=2X_H^{t-1}.
    \]
    By the condition of Theorem~\ref{mainthm} such that $\eta\leq 1/(2\tau)$, we have that $\eta\tau\leq1/2$. Therefore, we have
    \[
        \alpha^{-1}
        =
        (1-\eta\tau)^{-1/2}
        \leq
        \sqrt2.
    \]
    Combining these two inequalities with the estimate on $X_H^t$, for $t\geq1$,
    \[
        \bar X_H^{t}
        \leq
        2X_H^{t-1}
        \leq
        2\sqrt{K_0}\alpha^{t-1}
        \leq
        2\sqrt{2K_0}\alpha^t.
    \]
    The case $t=0$ is covered by $\bar X_H^0\leq\sqrt{K_0}$. Hence, by taking a constant $C_H$ larger than $(1 +2\sqrt{2})\sqrt{K_0}$, e.g.,
    \[
        C_H:=5\sqrt{K_0},
    \]
    we have
    \[
        \delta_H^t+X_H^t+\bar X_H^t
        \leq
        C_H\alpha^t
        =
        C_H(t+q_H+2)^{q_H}\alpha^t,
    \]
    where the equality holds because $q_H=0$.
    Thus the lemma holds at $h=H$ and for all $t\geq 0$.

    \paragraph{Case of $h<H$.}
    Now fix $h\leq H-1$ and suppose that the lemma holds at level $h+1$.
    In this part, we temporarily set
    \[
        q:=q_h=H-h.
    \]
    Then
    \[
        q\geq1,
        \qquad
        q_{h+1}=q-1.
    \]

    \paragraph{(1) Bounding $\delta_h^t$.}
    For the current $h$, we consider $\delta_h^t$ for every $t\geq 0$ simultaneously.

    Plugging the induction hypothesis
    \[
        \delta_{h+1}^t+X_{h+1}^t+\bar X_{h+1}^t
        \leq
        C_{h+1}\alpha^t
        =
        C_{h+1}(t+q_{h+1}+2)^{q_{h+1}}\alpha^t
    \]
    into the result of Lemma~\ref{iter2} at level $h$, we have
    \[
    \begin{aligned}
        \delta_h^t
        &\leq
        \eta\tau
        \sum_{k=0}^{t-1}
        \beta^k
        \left(
            \delta_{h+1}^{t-k}
            +
            NH\sqrt{2N}\,\bar X_{h+1}^{t-k}
        \right)
        +
        NH\beta^t
        \\
        &\leq
        \eta\tau
        \sum_{k=0}^{t-1}
        \beta^k
            (1+
            NH\sqrt{2N})(\delta_{h+1}^{t-k}
            +\bar X_{h+1}^{t-k})
        +
        NH\beta^t
        \\
        &\leq
        \eta\tau
        (1+NH\sqrt{2N})C_{h+1}
        \sum_{k=0}^{t-1}
        \beta^k
        (t-k+(q-1)+2)^{q-1}
        \alpha^{t-k}
        +
        NH\beta^t 
        \\
        &\leq
        \eta\tau
        (1+NH\sqrt{2N})C_{h+1}
        \sum_{k=0}^{t-1}
        \beta^k
        (t+q+1)^{q-1}
        \alpha^{t-k}
        +
        NH\alpha^t
        \\
        &=
        (1+NH\sqrt{2N})C_{h+1}(t+q+2)^{q-1}\cdot \eta\tau
        \sum_{k=0}^{t-1}
        \beta^k
        \alpha^{t-k}
        +
        NH\alpha^t.
    \end{aligned}
    \]
    where for the last term we used \(\beta = \alpha^2 \leq \alpha\).

    Since
    \[
        \beta=\alpha^2,
    \]
    we have
    \[
        \beta^k\alpha^{t-k}
        =
        \alpha^t\alpha^k.
    \]
    Hence,
    \[
        \eta\tau
        \sum_{k=0}^{t-1}
        \beta^k
        \alpha^{t-k}=\eta\tau\alpha^t
        \sum_{k=0}^{t-1}\alpha^k
        \leq
        \eta\tau\alpha^t
        \sum_{k=0}^{\infty}\alpha^k
        =
        \frac{\eta\tau\alpha^t}{1-\alpha}
        =
        (1+\alpha)\alpha^t
        \leq
        2\alpha^t,
    \]
    where we used $1-\eta\tau = \alpha^2$ and $\alpha \leq 1$ in the last two steps.

    Therefore, by defining 
    \[
        D_h:=NH+2(1+NH\sqrt{2N})C_{h+1},
    \]
    we obtain
    \[
    \begin{aligned}
        \delta_h^t
        &\leq
        \left[
            NH
            +
            2(1+NH\sqrt{2N})C_{h+1}
        \right]
        (t+q+2)^{q-1}\alpha^t
        \\
        &=
        D_h(t+q_h+2)^{q_h-1}\alpha^t.
    \end{aligned}
    \]
    This proves the bound for the first term $\delta_h^t$ of every $t$ at level $h$.

    \paragraph{(2) Bounding $X_h^t$.}
    We next prove the bound for $X_h^t$. Expanding the recursion
    \[
        X_h^{t+1}
        \leq
        \alpha X_h^t
        +
        \omega_\eta
        \left(
            \delta_h^t+2\delta_h^{t+1}
        \right)
    \]
    in Lemma~\ref{iter1} gives
    \begin{equation}
        \label{spanx}
        X_h^t
        \leq
        \alpha^tX_h^0
        +
        \omega_\eta
        \sum_{m=0}^{t-1}
        \alpha^{t-1-m}
        \left(
            \delta_h^m+2\delta_h^{m+1}
        \right).
    \end{equation}
    Using the newly proved bound on $\delta_h^t$, for every $m\geq0$,
    \[
        \delta_h^m
        \leq
        D_h(m+q+2)^{q-1}\alpha^m\leq
        D_h(m+q+3)^{q-1}\alpha^m,
    \]
    and
    \[
        \delta_h^{m+1}
        \leq
        D_h(m+q+3)^{q-1}\alpha^{m+1}
        \leq
        D_h(m+q+3)^{q-1}\alpha^m.
    \]
    Hence
    \[
        \delta_h^m+2\delta_h^{m+1}
        \leq
        3D_h(m+q+3)^{q-1}\alpha^m.
    \]
    Putting this into Eq.~(\ref{spanx}) gives
    \[
    \begin{aligned}
        X_h^t
        &\leq
        \sqrt{K_0}\alpha^t
        +
        3\omega_\eta D_h
        \sum_{m=0}^{t-1}
        \alpha^{t-1-m}
        (m+q+3)^{q-1}
        \alpha^m
        \\
        &=
        \sqrt{K_0}\alpha^t
        +
        3\omega_\eta D_h
        \alpha^{t-1}
        \sum_{m=0}^{t-1}
        (m+q+3)^{q-1}.
    \end{aligned}
    \]
    Since $m\leq t-1$ implies
    \[
        m+q+3\leq t+q+2,
    \]
    we have
    \[
        \sum_{m=0}^{t-1}
        (m+q+3)^{q-1}
        \leq
        t(t+q+2)^{q-1}
        \leq
        (t+q+2)^q.
    \]
    Also $\alpha^{-1}\leq\sqrt2\leq2$. Thus
    \[
    \begin{aligned}
        X_h^t
        &\leq
        \sqrt{K_0}\alpha^t
        +
        3\omega_\eta D_h
        \alpha^{t-1}
        (t+q+2)^q
        \\
        &\leq
        \left[
            \sqrt{K_0}
            +
            6\omega_\eta D_h
        \right]
        (t+q+2)^q\alpha^t.
    \end{aligned}
    \]
    Define
    \[
        A_h:=
        \sqrt{K_0}
        +
        6\omega_\eta D_h.
    \]
    Then
    \[
        X_h^t
        \leq
        A_h(t+q_h+2)^{q_h}\alpha^t.
    \]

    \paragraph{(3) Bounding $\bar X_h^t$.}
    Finally, we prove the bound for $\bar X_h^t$. For $t\geq1$, Lemma~\ref{iter1} gives
    \[
    \begin{aligned}
        \bar X_h^t
        &\leq
        2X_h^{t-1}
        +
        \omega_\eta
        \left(
            \delta_h^{t-1}+2\delta_h^t
        \right).
    \end{aligned}
    \]
    The first term is bounded as
    \[
        2X_h^{t-1}
        \leq
        2A_h(t+q+1)^q\alpha^{t-1}
        \leq
        4A_h(t+q+2)^q\alpha^t,
    \]
    where we used $\alpha^{-1}\leq\sqrt2\leq2$.
    For the second term, using the bound on $\delta_h^t$,
    \[
    \begin{aligned}
        \delta_h^{t-1}+2\delta_h^t
        &\leq
        D_h(t+q+1)^{q-1}\alpha^{t-1}
        +
        2D_h(t+q+2)^{q-1}\alpha^t
        \\
        &\leq
        3D_h(t+q+2)^{q-1}\alpha^{t-1}
        \\
        &\leq
        6D_h(t+q+2)^q\alpha^t.
    \end{aligned}
    \]
    Hence
    \[
        \bar X_h^t
        \leq
        \left[
            4A_h
            +
            6\omega_\eta D_h
        \right]
        (t+q+2)^q\alpha^t.
    \]
    The case $t=0$ is covered by $\bar X_h^0\leq\sqrt{K_0}$ and the definition of $A_h$.

    \paragraph{(4) Integrating the bounds.}

    Define a constant 
    \[
        C_h=
        5\sqrt{K_0}+(1+36\omega_\eta) D_h.
    \]

    Combining the three bounds, for all $t\geq0$, we have
    \[
    \begin{aligned}
        \delta_h^t+X_h^t+\bar X_h^t
        &\leq
        \left[
            D_h
            +
            A_h
            +
            4A_h
            +
            6\omega_\eta D_h
        \right]
        (t+q+2)^q\alpha^t
        \\
        &=
        [
            5\sqrt{K_0}+(1+36\omega_\eta) D_h
        ]
        (t+q+2)^q\alpha^t
        \\
        &=
        C_h(t+q+2)^q\alpha^t.
    \end{aligned}
    \]
    Since $q=q_h$, the induction step is complete. 
    
    By backward induction, we finally obtain the result that for every $h\in[H]$ and $t\geq0$,
    \[
        \delta_h^t+\sqrt{L_h^t}+
        \sqrt{\bar L_h^t}
        \leq
        C_h(t+H-h+2)^{H-h}(1-\eta\tau)^{t/2}.
    \]

    \paragraph{Bounding the constant $C_h$.}
    Based on the former deduction, we have
    \[
        C_H=5\sqrt{K_0},\quad D_H=0,
    \]
    Then for $h<H$, we have
    \[
        D_h:=NH+2(1+NH\sqrt{2N})C_{h+1},
    \]
    \[
        C_h=
        5\sqrt{K_0}+(1+36\omega_\eta) D_h.
    \]
    The two equations give
    \[
    \begin{split}
        C_h
        &=
        5\sqrt{K_0}+(1+36\omega_\eta)(NH+2(1+NH\sqrt{2N})C_{h+1})
        \\
        &=5\sqrt{K_0}+(1+36\omega_\eta)NH+2(1+36\omega_\eta)(1+NH\sqrt{2N})C_{h+1}.
    \end{split}
    \]
    Recall that 
    \[
    \omega_\eta=\sqrt{256N^2\eta(1+1/\tau)}=16N\sqrt{\eta(1+1/\tau)}\leq O(N/\tau).
    \]
    and \(K_0=N\log A\). Hence, by taking a sufficiently large $C$,we have
    \[
        1+C_h 
        \leq C N^3H\log A/\tau \cdot (1+ C_{h+1})
    \]
    By backward induction, we have
    \[
        C_h= O(N^3H\log A/\tau)^{H-h}.
    \]
\end{proof}

\subsection{Proof of Lemma~\ref{iter1}}

\begin{proof}
    \leavevmode\par 
    \paragraph{Auxiliary notations.}
    Fix $h\in[H]$ and $s\in\gS$. Throughout the proof, all policies are evaluated at the same pair $(h,s)$ unless explicitly stated otherwise. To simplify notation, we write
    \[
        \pi^*:=\pi_h^{*,\tau}(s),
        \qquad
        \pi^t:=\pi_h^t(s),
        \qquad
        \bar\pi^t:=\bar\pi_h^t(s),
        \qquad
        \bar\pi^{t+1}:=\bar\pi_h^{t+1}(s).
    \]
    KL divergences between product policies are understood as sums over players. For example,
    \[
        \KL(\pi^*||\pi^t)
        =
        \sum_{i=1}^{N}
        \KL(\pi_i^*||\pi_i^t).
    \]
    We also define
    \[
        \Delta_h^t
        :=
        \delta_h^t+2\delta_h^{t+1}.
    \]
    \paragraph{Roadmap of this proof.}
    We first derive a one-step identity that decomposes \(\KL(\pi_h^{*,\tau}||\pi_h^{t+1})(s)\) into several terms for the convenience of further analysis.
    \[
        \boxed{
        \begin{aligned}
            &
            \KL(\pi_h^{*,\tau}||\pi_h^{t+1})(s)
            \\
            &=
            \beta
            \KL(\pi_h^{*,\tau}||\pi_h^t)(s)
            -
            \beta
            \KL(\bar\pi_h^{t+1}||\pi_h^t)(s)
            -
            \KL(\pi_h^{t+1}||\bar\pi_h^{t+1})(s)
            \\
            &\quad
            -
            \eta\tau
            \KL(\bar\pi_h^{t+1}||\pi_h^{*,\tau})(s)
            +
            \left\langle
                \log\bar\pi_h^{t+1}
                -
                \log\pi_h^{t+1},
                \bar\pi_h^{t+1}
                -
                \pi_h^{t+1}
            \right\rangle(s)
            +
            \mathcal E_1 .
        \end{aligned}
        }
    \]

    Then in step 2 and step 3, we provide two bounds on \(\mathcal E_1\) and \(\left\langle
    \log\bar\pi_h^{t+1}
    -
    \log\pi_h^{t+1},
    \bar\pi_h^{t+1}
    -
    \pi_h^{t+1}
    \right\rangle\), respectively. From step 4 to 7, we will combine these bounds with the one-step identity to formally derive the policy-divergence recursions. We will box the results of each step for better readability.

    \paragraph{Step 1: a basic one-step identity.}
    To start with, from the OMWU update and the QRE first-order condition, for every player $i\in[N]$,
    \[
        \log \pi_{h,i}^{t+1}(s)
        \overset{\boldsymbol{1}}{=}
        \beta \log \pi_{h,i}^{t}(s)
        +
        \eta Q_{h,i}^{t+1}(s,\cdot,\bar\pi_{h,-i}^{t+1}),
    \]
    and
    \[
        \eta\tau \log \pi_{h,i}^{*,\tau}(s)
        \overset{\boldsymbol{1}}{=}
        \eta Q_{h,i}^{*,\tau}(s,\cdot,\pi_{h,-i}^{*,\tau}),
    \]
    where
    \[
        \beta=1-\eta\tau,
    \]
    and $x\overset{\boldsymbol{1}}{=}y$ means that $x-y$ is a constant vector.

    Subtracting the QRE condition from the OMWU update and pairing with
    \[
        \bar\pi_{h,i}^{t+1}(s)-\pi_{h,i}^{*,\tau}(s)
    \]
    gives
    \[
    \begin{aligned}
        &
        \left\langle
            \log \pi_{h,i}^{t+1}
            -
            \beta\log \pi_{h,i}^{t}
            -
            \eta\tau\log \pi_{h,i}^{*,\tau},
            \bar\pi_{h,i}^{t+1}
            -
            \pi_{h,i}^{*,\tau}
        \right\rangle(s)
        \\
        &=
        \eta
        \left\langle
            Q_{h,i}^{t+1}(s,\cdot,\bar\pi_{h,-i}^{t+1})
            -
            Q_{h,i}^{*,\tau}(s,\cdot,\pi_{h,-i}^{*,\tau}),
            \bar\pi_{h,i}^{t+1}
            -
            \pi_{h,i}^{*,\tau}
        \right\rangle(s).
    \end{aligned}
    \]
    Decompose the right-hand side as
    \[
    \begin{aligned}
        &
        \eta
        \left\langle
            Q_{h,i}^{t+1}(s,\cdot,\bar\pi_{h,-i}^{t+1})
            -
            Q_{h,i}^{*,\tau}(s,\cdot,\pi_{h,-i}^{*,\tau}),
            \bar\pi_{h,i}^{t+1}
            -
            \pi_{h,i}^{*,\tau}
        \right\rangle
        \\
        &=
        \eta
        \left\langle
            \left(Q_{h,i}^{t+1}-Q_{h,i}^{*,\tau}\right)
            (s,\cdot,\bar\pi_{h,-i}^{t+1}),
            \bar\pi_{h,i}^{t+1}
            -
            \pi_{h,i}^{*,\tau}
        \right\rangle
        \\
        &\quad+
        \eta
        Q_{h,i}^{*,\tau}
        \left(
            s,
            \bar\pi_{h,i}^{t+1}
            -
            \pi_{h,i}^{*,\tau},
            \bar\pi_{h,-i}^{t+1}
            -
            \pi_{h,-i}^{*,\tau}
        \right).
    \end{aligned}
    \]
    By Lemma~\ref{lem:qre-induced-structure}, $Q_h^{*,\tau}(s)$ is a zero-sum networked separable normal-form game. Therefore Lemma~\ref{sumzero2} gives
    \[
        \sum_{i=1}^{N}
        Q_{h,i}^{*,\tau}
        \left(
            s,
            \bar\pi_{h,i}^{t+1}
            -
            \pi_{h,i}^{*,\tau},
            \bar\pi_{h,-i}^{t+1}
            -
            \pi_{h,-i}^{*,\tau}
        \right)
        =
        0.
    \]
    Therefore, after summing over all players,
    \[
    \begin{aligned}
        &
        \left\langle
            \log \pi_h^{t+1}
            -
            \beta\log \pi_h^t
            -
            \eta\tau\log \pi_h^{*,\tau},
            \bar\pi_h^{t+1}
            -
            \pi_h^{*,\tau}
        \right\rangle(s)
        \\
        &=
        \eta
        \sum_{i=1}^{N}
        \left\langle
            \left(Q_{h,i}^{t+1}-Q_{h,i}^{*,\tau}\right)
            (s,\cdot,\bar\pi_{h,-i}^{t+1}),
            \bar\pi_{h,i}^{t+1}
            -
            \pi_{h,i}^{*,\tau}
        \right\rangle(s).
    \end{aligned}
    \]

We denote the right-hand side by
\[
\begin{aligned}
    \mathcal E_1
    &:=
    \eta
    \sum_{i=1}^{N}
    \left\langle
        \left(
            Q_{h,i}^{t+1}
            -
            Q_{h,i}^{*,\tau}
        \right)
        (s,\cdot,\bar\pi_{h,-i}^{t+1}),
        \bar\pi_{h,i}^{t+1}(s)
        -
        \pi_{h,i}^{*,\tau}(s)
    \right\rangle .
\end{aligned}
\]

It remains to rewrite the left-hand side in terms of KL divergences. For notational clarity, fix $h,t,s$ and write
\[
    x=\pi_h^t,\qquad
    y=\pi_h^{t+1},\qquad
    z=\bar\pi_h^{t+1},\qquad
    p=\pi_h^{*,\tau}.
\]
Define
\[
    \Phi
    :=
    \log y-\beta\log x-\eta\tau\log p .
\]
Since $\beta+\eta\tau=1$, we first compute
\[
\begin{aligned}
    \langle \Phi,z\rangle(s)
    &=
    \left\langle
        \log y-\beta\log x-\eta\tau\log p,
        z
    \right\rangle(s)
    \\
    &=
    \left\langle
        \log z-\beta\log x-\eta\tau\log p,
        z
    \right\rangle(s)
    +
    \left\langle
        \log y-\log z,
        z
    \right\rangle(s)
    \\
    &=
    \left[
        \beta\KL(z||x)
        +
        \eta\tau\KL(z||p)
    \right](s)
    +
    \left\langle
        \log y-\log z,
        z
    \right\rangle(s).
\end{aligned}
\]
The last inner product is further decomposed as
\[
\begin{aligned}
    \left\langle
        \log y-\log z,
        z
    \right\rangle(s)
    &=
    \left\langle
        \log y-\log z,
        y
    \right\rangle(s)
    +
    \left\langle
        \log y-\log z,
        z-y
    \right\rangle(s)
    \\
    &=
    \KL(y||z)(s)
    -
    \left\langle
        \log z-\log y,
        z-y
    \right\rangle(s).
\end{aligned}
\]
Therefore,
\[
\begin{aligned}
    \langle \Phi,z\rangle(s)
    =
    \left[
        \beta\KL(z||x)
        +
        \eta\tau\KL(z||p)
        +
        \KL(y||z)
    \right](s)
    -
    \left\langle
        \log z-\log y,
        z-y
    \right\rangle(s).
\end{aligned}
\]
Next,
\[
\begin{aligned}
    -\langle \Phi,p\rangle(s)
    &=
    -\left\langle
        \log y-\beta\log x-\eta\tau\log p,
        p
    \right\rangle(s)
    \\
    &=
    \left[
        \KL(p||y)
        -
        \beta\KL(p||x)
    \right](s).
\end{aligned}
\]
Combining the two displayed identities gives
\[
\begin{aligned}
    \left\langle
        \Phi,
        z-p
    \right\rangle(s)
    &=
    \langle\Phi,z\rangle(s)-\langle\Phi,p\rangle(s)
    \\
    &=
    \left[
        \KL(p||y)
        -
        \beta\KL(p||x)
        +
        \beta\KL(z||x)
        +
        \KL(y||z)
        +
        \eta\tau\KL(z||p)
    \right](s)
    \\
    &\quad
    -
    \left\langle
        \log z-\log y,
        z-y
    \right\rangle(s).
\end{aligned}
\]
Substituting back
$x=\pi_h^t$,
$y=\pi_h^{t+1}$,
$z=\bar\pi_h^{t+1}$,
and
$p=\pi_h^{*,\tau}$,
and using
$\langle\Phi,z-p\rangle(s)=\mathcal E_1$,
we obtain
\[
    \begin{aligned}
        &
        \KL(\pi_h^{*,\tau}||\pi_h^{t+1})(s)
        \\
        &=
        \beta
        \KL(\pi_h^{*,\tau}||\pi_h^t)(s)
        -
        \beta
        \KL(\bar\pi_h^{t+1}||\pi_h^t)(s)
        -
        \KL(\pi_h^{t+1}||\bar\pi_h^{t+1})(s)
        \\
        &\quad
        -
        \eta\tau
        \KL(\bar\pi_h^{t+1}||\pi_h^{*,\tau})(s)
        +
        \left\langle
            \log\bar\pi_h^{t+1}
            -
            \log\pi_h^{t+1},
            \bar\pi_h^{t+1}
            -
            \pi_h^{t+1}
        \right\rangle(s)
        +
        \mathcal E_1 .
    \end{aligned}
\]

    \paragraph{Step 2: a quadratic bound on $\mathcal E_1$ term.}
    We show that
    \[
    \boxed{
    \begin{aligned}
        \mathcal E_1
        &\leq
        \frac{\eta\tau}{4}
        \KL(
            \bar\pi_h^{t+1}
            ||
            \pi_h^{*,\tau}
        )
        +
        \frac{2\eta N}{\tau}
        (\delta_h^{t+1})^2.
    \end{aligned}
    }
    \]
    By the definition
    $$\delta_h^{t+1}=
    \max_{s\in\gS}
    \max_{i\in[N]}
    \max_{\va\in\gA}
    \left|
        Q_{h,i}^{t+1}(s,\va)-Q_{h,i}^{*,\tau}(s,\va)
    \right|,$$ we have
    \[
    \begin{aligned}
        \mathcal E_1&=
    \eta
    \sum_{i=1}^{N}
    \left\langle
        \left(
            Q_{h,i}^{t+1}
            -
            Q_{h,i}^{*,\tau}
        \right)
        (s,\cdot,\bar\pi_{h,-i}^{t+1}),
        \bar\pi_{h,i}^{t+1}(s)
        -
        \pi_{h,i}^{*,\tau}(s)
    \right\rangle\\
        &\leq
        \eta
        \sum_{i=1}^{N}\delta_h^{t+1}
        \left\|
            \pi_{h,i}^{*,\tau}(\cdot|s)
            -
            \bar\pi_{h,i}^{t+1}(\cdot|s)
        \right\|_1\\
        &=
        \eta\delta_h^{t+1}
        \sum_{i=1}^{N}
        \left\|
            \pi_{h,i}^{*,\tau}(\cdot|s)
            -
            \bar\pi_{h,i}^{t+1}(\cdot|s)
        \right\|_1.
    \end{aligned}
    \]
    Applying Cauchy-Schwarz inequality and Pinsker's inequality in Lemma~\ref{lem:pinsker},
    \[
    \begin{aligned}
        \sum_{i=1}^{N}
        \left\|
            \pi_{h,i}^{*,\tau}
            -
            \bar\pi_{h,i}^{t+1}
        \right\|_1
        &\leq
        \sqrt{
            N
            \sum_{i=1}^{N}
            \left\|
                \pi_{h,i}^{*,\tau}
                -
                \bar\pi_{h,i}^{t+1}
            \right\|_1^2
        }
        \\
        &\leq
        \sqrt{
            2N
            \KL(
                \bar\pi_h^{t+1}
                ||
                \pi_h^{*,\tau}
            )
        }.
    \end{aligned}
    \]
    Hence, by Young's inequality,
    \[
    \begin{aligned}
        \mathcal E_1
        &\leq
        \eta\delta_h^{t+1}
        \sqrt{
            2N
            \KL(
                \bar\pi_h^{t+1}
                ||
                \pi_h^{*,\tau}
            )
        }
        \\
        &\leq
        \frac{\eta\tau}{4}
        \KL(
            \bar\pi_h^{t+1}
            ||
            \pi_h^{*,\tau}
        )
        +
        \frac{2\eta N}{\tau}
        (\delta_h^{t+1})^2.
    \end{aligned}
    \]

    \paragraph{Step 3: controlling the optimistic correction term.}
    In this step, we prove that
    \[
    \boxed{
    \begin{aligned}
        &
        \left\langle
            \log\bar\pi_h^{t+1}
            -
            \log\pi_h^{t+1},
            \bar\pi_h^{t+1}
            -
            \pi_h^{t+1}
        \right\rangle(s)
        \\
        &\leq
        \eta N^2H
        \KL(\pi_h^t||\bar\pi_h^t)(s)
        +
        \eta N^2H
        \KL(\bar\pi_h^{t+1}||\pi_h^t)(s)
        \\
        &\quad+
        3\eta N^2H
        \KL(\pi_h^{t+1}||\bar\pi_h^{t+1})(s)+
        4\eta N^2
        (\delta_h^t+\delta_h^{t+1})^2.
    \end{aligned}
    }
    \]
    From the optimistic update in the algorithm,
    \begin{equation}
        \label{optupdate}
        \log \bar\pi_{h,i}^{t+1}(s)
        -
        \log \pi_{h,i}^{t+1}(s)
        \overset{\boldsymbol{1}}{=}
        \eta Q_{h,i}^{t}
        (s,\cdot,\bar\pi_{h,-i}^{t})
        -
        \eta Q_{h,i}^{t+1}
        (s,\cdot,\bar\pi_{h,-i}^{t+1}).
    \end{equation}
    Therefore,
    \[
    \begin{aligned}
        &
        \left\langle
            \log\bar\pi_h^{t+1}
            -
            \log\pi_h^{t+1},
            \bar\pi_h^{t+1}
            -
            \pi_h^{t+1}
        \right\rangle(s)
        \\
        &=
        \eta
        \sum_{i=1}^{N}
        \left\langle
            \bar\pi_{h,i}^{t+1}
            -
            \pi_{h,i}^{t+1},
            Q_{h,i}^{t}
            (s,\cdot,\bar\pi_{h,-i}^{t})
            -
            Q_{h,i}^{t+1}
            (s,\cdot,\bar\pi_{h,-i}^{t+1})
        \right\rangle(s).
    \end{aligned}
    \]
    Insert and subtract $Q_{h,i}^{*,\tau}$ to decompose this term into
    \[
    \begin{aligned}
        &
        \left\langle
            \log\bar\pi_h^{t+1}
            -
            \log\pi_h^{t+1},
            \bar\pi_h^{t+1}
            -
            \pi_h^{t+1}
        \right\rangle(s)
        \\
        &=
        \eta
        \sum_{i=1}^{N}
        \left\langle
            \bar\pi_{h,i}^{t+1}
            -
            \pi_{h,i}^{t+1},
            Q_{h,i}^{*,\tau}
            (s,\cdot,\bar\pi_{h,-i}^{t})
            -
            Q_{h,i}^{*,\tau}
            (s,\cdot,\bar\pi_{h,-i}^{t+1})
        \right\rangle(s)
        \\
        &\quad+
        \mathcal E_2,
    \end{aligned}
    \]
    where
    \[
    \begin{aligned}
        \mathcal E_2
        &:=
        \eta
        \sum_{i=1}^{N}
        \left\langle
            \bar\pi_{h,i}^{t+1}
            -
            \pi_{h,i}^{t+1},
            \left(Q_{h,i}^{t}-Q_{h,i}^{*,\tau}\right)
            (s,\cdot,\bar\pi_{h,-i}^{t})
        \right\rangle
        \\
        &\quad+
        \eta
        \sum_{i=1}^{N}
        \left\langle
            \bar\pi_{h,i}^{t+1}
            -
            \pi_{h,i}^{t+1},
            \left(Q_{h,i}^{*,\tau}-Q_{h,i}^{t+1}\right)
            (s,\cdot,\bar\pi_{h,-i}^{t+1})
        \right\rangle .
    \end{aligned}
    \]
    By the definitions of $\delta_h^t$ and $\delta_h^{t+1}$,
    \[
    \begin{aligned}
        \mathcal E_2
        &\leq
        \eta(\delta_h^t+\delta_h^{t+1})
        \sum_{i=1}^{N}
        \left\|
            \bar\pi_{h,i}^{t+1}
            -
            \pi_{h,i}^{t+1}
        \right\|_1.
    \end{aligned}
    \]
    Again by Cauchy-Schwarz, Pinsker, and Young's inequality,
    \[
    \begin{aligned}
        \mathcal E_2
        &\leq
        \eta(\delta_h^t+\delta_h^{t+1})
        \sqrt{
            2N
            \KL(
                \pi_h^{t+1}
                ||
                \bar\pi_h^{t+1}
            )
        }
        \\
        &\leq \eta N
        \KL(
            \pi_h^{t+1}
            ||
            \bar\pi_h^{t+1}
        )
        +
        \frac{1}{2}\eta 
        (\delta_h^t+\delta_h^{t+1})^2\\
        &\leq
        \eta N^2H
        \KL(
            \pi_h^{t+1}
            ||
            \bar\pi_h^{t+1}
        )
        +
        4\eta N^2
        (\delta_h^t+\delta_h^{t+1})^2.
    \end{aligned}
    \]

    It remains to control the term involving $Q_h^{*,\tau}$. By Lemma~\ref{lem:qre-induced-structure}, $Q_h^{*,\tau}(s)$ is a networked separable normal-form game. Moreover $\|Q_{h,i,j}^{*,\tau}\|_\infty\leq NH$ by Eq.~(\ref{eq:edge-q-bound}), in both controller cases. Hence
    \[
    \begin{aligned}
        &
        \sum_{i=1}^{N}
        \left\langle
            \bar\pi_{h,i}^{t+1}
            -
            \pi_{h,i}^{t+1},
            Q_{h,i}^{*,\tau}
            (s,\cdot,\bar\pi_{h,-i}^{t})
            -
            Q_{h,i}^{*,\tau}
            (s,\cdot,\bar\pi_{h,-i}^{t+1})
        \right\rangle
        \\
        &=
        \sum_{i=1}^{N}
        \sum_{j\in\gN_i}
        \left(
            \bar\pi_{h,i}^{t+1}
            -
            \pi_{h,i}^{t+1}
        \right)^\top
        Q_{h,i,j}^{*,\tau}(s,\cdot,\cdot)
        \left(
            \bar\pi_{h,j}^{t}
            -
            \bar\pi_{h,j}^{t+1}
        \right)
        \\
        &=
        \sum_{i=1}^{N}
        \sum_{j\in\gN_i}
        \left(
            \bar\pi_{h,i}^{t+1}
            -
            \pi_{h,i}^{t+1}
        \right)^\top
        Q_{h,i,j}^{*,\tau}(s,\cdot,\cdot)
        \left(
            \bar\pi_{h,j}^{t}
            -
            \pi_{h,j}^{t}
        \right)
        \\
        &\quad+
        \sum_{i=1}^{N}
        \sum_{j\in\gN_i}
        \left(
            \bar\pi_{h,i}^{t+1}
            -
            \pi_{h,i}^{t+1}
        \right)^\top
        Q_{h,i,j}^{*,\tau}(s,\cdot,\cdot)
        \left(
            \pi_{h,j}^{t}
            -
            \bar\pi_{h,j}^{t+1}
        \right).
    \end{aligned}
    \]
    Using the fact that $\|Q_{h,i,j}^{*,\tau}\|_\infty\leq NH$ for all $i,j$, we have that
    \[
        \left(
            \bar\pi_{h,i}^{t+1}
            -
            \pi_{h,i}^{t+1}
        \right)^\top
        Q_{h,i,j}^{*,\tau}(s,\cdot,\cdot)
        \left(
            \bar\pi_{h,j}^{t}
            -
            \pi_{h,j}^{t}
        \right)\leq \|
            \bar\pi_{h,i}^{t+1}
            -
            \pi_{h,i}^{t+1}
        \|_1
        \|Q_{h,i,j}^{*,\tau}(s,\cdot,\cdot)\|_{\infty}
        \|
            \bar\pi_{h,j}^{t}
            -
            \pi_{h,j}^{t}
        \|_1
    \]
    Hence, one can bound the following expression as
    \[
    \begin{aligned}
        &
        \sum_{i=1}^{N}
        \left\langle
            \bar\pi_{h,i}^{t+1}
            -
            \pi_{h,i}^{t+1},
            Q_{h,i}^{*,\tau}
            (s,\cdot,\bar\pi_{h,-i}^{t})
            -
            Q_{h,i}^{*,\tau}
            (s,\cdot,\bar\pi_{h,-i}^{t+1})
        \right\rangle
        \\
        &\leq
        NH
        \sum_{i=1}^{N}
        \sum_{j=1}^{N}
        \left\|
            \bar\pi_{h,i}^{t+1}
            -
            \pi_{h,i}^{t+1}
        \right\|_1
        \left\|
            \bar\pi_{h,j}^{t}
            -
            \pi_{h,j}^{t}
        \right\|_1
        \\
        &\quad+
        NH
        \sum_{i=1}^{N}
        \sum_{j=1}^{N}
        \left\|
            \bar\pi_{h,i}^{t+1}
            -
            \pi_{h,i}^{t+1}
        \right\|_1
        \left\|
            \pi_{h,j}^{t}
            -
            \bar\pi_{h,j}^{t+1}
        \right\|_1
        \\
        &\leq
        NH
        \sum_{i=1}^{N}
        \sum_{j=1}^{N}
        \frac12
        \Bigg[
            2
            \left\|
                \bar\pi_{h,i}^{t+1}
                -
                \pi_{h,i}^{t+1}
            \right\|_1^2
            +
            \left\|
                \bar\pi_{h,j}^{t}
                -
                \pi_{h,j}^{t}
            \right\|_1^2
            \\
        &\hspace{8em}
            +
            \left\|
                \pi_{h,j}^{t}
                -
                \bar\pi_{h,j}^{t+1}
            \right\|_1^2
        \Bigg],
    \end{aligned}
    \]
    where in the last step, we apply \(ab\leq \frac{a^2+b^2}{2}\).

    Applying Pinsker's inequality in Lemma~\ref{lem:pinsker} to each squared total variation distance, we get
\begin{equation}
    \label{tmpeq2}
    \begin{aligned}
        &
        \sum_{i=1}^{N}
        \left\langle
            \bar\pi_{h,i}^{t+1}
            -
            \pi_{h,i}^{t+1},
            Q_{h,i}^{*,\tau}
            (s,\cdot,\bar\pi_{h,-i}^{t})
            -
            Q_{h,i}^{*,\tau}
            (s,\cdot,\bar\pi_{h,-i}^{t+1})
        \right\rangle
        \\
        &\leq
        N^2H
        \KL(\pi_h^t||\bar\pi_h^t)(s)
        +
        N^2H
        \KL(\bar\pi_h^{t+1}||\pi_h^t)(s)
        \\
        &\quad+
        2N^2H
        \KL(\pi_h^{t+1}||\bar\pi_h^{t+1})(s).
    \end{aligned}
\end{equation}
    Combining this with the bound on $\mathcal E_2$, we obtain
    \[
    \begin{aligned}
        &
        \left\langle
            \log\bar\pi_h^{t+1}
            -
            \log\pi_h^{t+1},
            \bar\pi_h^{t+1}
            -
            \pi_h^{t+1}
        \right\rangle(s)
        \\
        &\leq
        \eta N^2H
        \KL(\pi_h^t||\bar\pi_h^t)(s)
        +
        \eta N^2H
        \KL(\bar\pi_h^{t+1}||\pi_h^t)(s)
        \\
        &\quad+
        3\eta N^2H
        \KL(\pi_h^{t+1}||\bar\pi_h^{t+1})(s)
        \\
        &\quad+
        4\eta N^2
        (\delta_h^t+\delta_h^{t+1})^2.
    \end{aligned}
    \]

    \paragraph{Step 4: proof of the recursion for $L_h^{t+1}$.}
    We first recall some notations: \[
    \beta=1-\eta\tau,
    \qquad
    \lambda_\eta=1-4\eta N^2H,
    \qquad
    \Gamma_\eta=
    256N^2\eta\left(1+1/\tau\right).
\]
    In this part, we formally prove that 
    \[
    \boxed{
        L_h^{t+1}
        \leq
        \beta L_h^t
        +
        \Gamma_\eta
        \left(
            \delta_h^t+2\delta_h^{t+1}
        \right)^2.
    }
    \]
    Combining the identity from Step 1 with the bounds from Steps 2 and 3, we get
\begin{equation}
    \label{tmpeq}
    \begin{aligned}
        &
        \KL(\pi_h^{*,\tau}||\pi_h^{t+1})(s)
        \\
        &\leq
        \beta
        \KL(\pi_h^{*,\tau}||\pi_h^t)(s)
        +
        \eta N^2H
        \KL(\pi_h^t||\bar\pi_h^t)(s)
        \\
        &\quad
        -
        (\beta-\eta N^2H)
        \KL(\bar\pi_h^{t+1}||\pi_h^t)(s)
        \\
        &\quad
        -
        (1-3\eta N^2H)
        \KL(\pi_h^{t+1}||\bar\pi_h^{t+1})(s)
        \\
        &\quad
        -
        \frac{3\eta\tau}{4}
        \KL(\bar\pi_h^{t+1}||\pi_h^{*,\tau})(s)
        \\
        &\quad+
        \frac{2\eta N}{\tau}
        (\delta_h^{t+1})^2
        +
        4\eta N^2
        (\delta_h^t+\delta_h^{t+1})^2.
    \end{aligned}
\end{equation}
    Since
    \[
        \eta\leq \min\left\{\frac{1}{16N^2H},\frac{1}{2\tau}\right\},
    \]
    we have
    \[
        \eta N^2H\leq \frac{1}{16},
    \]
    and
    \[
        \beta = 1-\eta\tau\geq \frac{1}{2} \geq \eta N^2 H.
    \]
    Therefore, \[\beta - \eta N^2 H\geq 0 \] is non-negative.

    Consequently, by dropping the first and third non-positive terms from the right-hand side, we have
    \begin{equation*}
    \begin{aligned}
        &
        \KL(\pi_h^{*,\tau}||\pi_h^{t+1})(s)
        \\
        &\leq
        \beta
        \KL(\pi_h^{*,\tau}||\pi_h^t)(s)
        +
        \eta N^2H
        \KL(\pi_h^t||\bar\pi_h^t)(s)
        \\
        &\quad
        -
        (1-3\eta N^2H)
        \KL(\pi_h^{t+1}||\bar\pi_h^{t+1})(s)
        \\
        &\quad+
        \frac{2\eta N}{\tau}
        (\delta_h^{t+1})^2
        +
        4\eta N^2
        (\delta_h^t+\delta_h^{t+1})^2.
    \end{aligned}
    \end{equation*}
    Using \(
        \lambda_\eta=1-4\eta N^2H\leq 1-3\eta N^2H,
    \) and rearranging terms yield
    \[
    \begin{aligned}
        &
        \KL(\pi_h^{*,\tau}||\pi_h^{t+1})(s)+
        \lambda_\eta
        \KL(\pi_h^{t+1}||\bar\pi_h^{t+1})(s)
        \\
        &\leq
        \KL(\pi_h^{*,\tau}||\pi_h^{t+1})(s)+
        (1-3\eta N^2H)
        \KL(\pi_h^{t+1}||\bar\pi_h^{t+1})(s)
        \\
        &\leq
        \beta
        \KL(\pi_h^{*,\tau}||\pi_h^t)(s)
        +
        \eta N^2H
        \KL(\pi_h^t||\bar\pi_h^t)(s)
        \\
        &\quad+
        \frac{2\eta N}{\tau}
        (\delta_h^{t+1})^2
        +
        4\eta N^2
        (\delta_h^t+\delta_h^{t+1})^2.
    \end{aligned}
    \]
    Noting that
    \[
        \beta\lambda_\eta = (1-\eta\tau)\lambda_\eta\geq \frac{1}{2}\lambda_\eta = \frac{1}{2}(1-4\eta N^2H)\geq \frac{1}{2}(1-\frac{1}{4})=\frac{6}{16}\geq\eta N^2H,
    \]
    we obtain
    \[
    \begin{aligned}
        &
        \KL(\pi_h^{*,\tau}||\pi_h^{t+1})(s)
        +
        \lambda_\eta
        \KL(\pi_h^{t+1}||\bar\pi_h^{t+1})(s)
        \\
        &\leq
        \beta
        \left[
            \KL(\pi_h^{*,\tau}||\pi_h^t)(s)
            +
            \lambda_\eta
            \KL(\pi_h^t||\bar\pi_h^t)(s)
        \right]
        \\
        &\quad+
        \frac{2\eta N}{\tau}
        (\delta_h^{t+1})^2
        +
        4\eta N^2
        (\delta_h^t+\delta_h^{t+1})^2.
    \end{aligned}
    \]
    By the definition
    \[
        \Delta_h^t=\delta_h^t+2\delta_h^{t+1},
    \]
    and since $N\leq N^2$, we have
    \[
        \frac{2\eta N}{\tau}
        (\delta_h^{t+1})^2
        +
        4\eta N^2
        (\delta_h^t+\delta_h^{t+1})^2
        \leq
        \Gamma_\eta
        (\Delta_h^t)^2,
    \]
    where
    \[
        \Gamma_\eta
        =
        256N^2\eta
        \left(1+\frac1\tau\right).
    \]
    Therefore,
    \[
    \begin{aligned}
        &
        \KL(\pi_h^{*,\tau}||\pi_h^{t+1})(s)
        +
        \lambda_\eta
        \KL(\pi_h^{t+1}||\bar\pi_h^{t+1})(s)
        \\
        &\leq
        \beta
        \left[
            \KL(\pi_h^{*,\tau}||\pi_h^t)(s)
            +
            \lambda_\eta
            \KL(\pi_h^t||\bar\pi_h^t)(s)
        \right]
        +
        \Gamma_\eta
        (\Delta_h^t)^2.
    \end{aligned}
    \]
    Taking the maximum over $s\in\gS$ yields
    \[
        L_h^{t+1}
        \leq
        \beta L_h^t
        +
        \Gamma_\eta
        \left(
            \delta_h^t+2\delta_h^{t+1}
        \right)^2.
    \]

    \paragraph{Step 5: the reverse-KL bound.}
    From the inequality (\ref{tmpeq}), we can move the reverse-KL term to the left-hand side and discard all other nonnegative terms. This gives
    \[
    \begin{aligned}
        \frac{3\eta\tau}{4}
        \KL(\bar\pi_h^{t+1}||\pi_h^{*,\tau})(s)
        &\leq
        \beta
        \KL(\pi_h^{*,\tau}||\pi_h^t)(s)
        +
        \eta N^2H
        \KL(\pi_h^t||\bar\pi_h^t)(s)
        \\
        &\quad+
        \frac{2\eta N}{\tau}
        (\delta_h^{t+1})^2
        +
        4\eta N^2
        (\delta_h^t+\delta_h^{t+1})^2.
    \end{aligned}
    \]
    Since $\eta N^2H\leq \beta\lambda_\eta$, the first two terms on the right-hand side are bounded by
    \[
        \beta
        \left[
            \KL(\pi_h^{*,\tau}||\pi_h^t)(s)
            +
            \lambda_\eta
            \KL(\pi_h^t||\bar\pi_h^t)(s)
        \right].
    \]
    Using again the definition of $\Gamma_\eta$, we obtain
    \[
    \boxed{
        \frac{\eta\tau}{2}
        \KL(\bar\pi_h^{t+1}(s)||\pi_h^{*,\tau}(s))
        \leq
        \beta L_h^t
        +
        \Gamma_\eta
        \left(
            \delta_h^t+2\delta_h^{t+1}
        \right)^2.
    }
    \]

    \paragraph{Step 6: proof of the recursion for $\bar L_h^{t+1}$.}
    In this part, we prove that
    \[
    \boxed{
        \bar L_h^{t+1}
        \leq
        4L_h^t
        +
        \Gamma_\eta
        \left(
            \delta_h^t+2\delta_h^{t+1}
        \right)^2.}
    \]
    We now bound
    \[
        \KL(\pi_h^{*,\tau}||\bar\pi_h^{t+1})(s).
    \]
    The KL identity
    \[
    \begin{aligned}
        \KL(\pi_h^{*,\tau}||\bar\pi_h^{t+1})(s)
        &=
        \KL(\pi_h^{*,\tau}||\pi_h^{t+1})(s)
        -
        \KL(\bar\pi_h^{t+1}||\pi_h^{t+1})(s)
        \\
        &\quad
        -
        \left\langle
            \pi_h^{*,\tau}
            -
            \bar\pi_h^{t+1},
            \log\bar\pi_h^{t+1}
            -
            \log\pi_h^{t+1}
        \right\rangle(s)
    \end{aligned}
    \]
    follows directly from expanding the three KL divergences.

    We bound the last inner product using the optimistic-update identity (\ref{optupdate}). Define
    \[
        \mathcal J
        :=
        -
        \left\langle
            \pi_h^{*,\tau}
            -
            \bar\pi_h^{t+1},
            \log\bar\pi_h^{t+1}
            -
            \log\pi_h^{t+1}
        \right\rangle(s).
    \]
    Then
    \[
    \begin{aligned}
        \mathcal J
        &=
        \eta
        \sum_{i=1}^{N}
        \left\langle
            \bar\pi_{h,i}^{t+1}
            -
            \pi_{h,i}^{*,\tau},
            Q_{h,i}^{t}
            (s,\cdot,\bar\pi_{h,-i}^{t})
            -
            Q_{h,i}^{t+1}
            (s,\cdot,\bar\pi_{h,-i}^{t+1})
        \right\rangle
        \\
        &=
        \eta
        \sum_{i=1}^{N}
        \left\langle
            \bar\pi_{h,i}^{t+1}
            -
            \pi_{h,i}^{*,\tau},
            Q_{h,i}^{*,\tau}
            (s,\cdot,\bar\pi_{h,-i}^{t})
            -
            Q_{h,i}^{*,\tau}
            (s,\cdot,\bar\pi_{h,-i}^{t+1})
        \right\rangle
        \\
        &\quad+
        \mathcal E_3,
    \end{aligned}
    \]
    where
    \[
    \begin{aligned}
        \mathcal E_3
        &:=
        \eta
        \sum_{i=1}^{N}
        \left\langle
            \bar\pi_{h,i}^{t+1}
            -
            \pi_{h,i}^{*,\tau},
            \left(Q_{h,i}^{t}-Q_{h,i}^{*,\tau}\right)
            (s,\cdot,\bar\pi_{h,-i}^{t})
        \right\rangle
        \\
        &\quad+
        \eta
        \sum_{i=1}^{N}
        \left\langle
            \bar\pi_{h,i}^{t+1}
            -
            \pi_{h,i}^{*,\tau},
            \left(Q_{h,i}^{*,\tau}-Q_{h,i}^{t+1}\right)
            (s,\cdot,\bar\pi_{h,-i}^{t+1})
        \right\rangle .
    \end{aligned}
    \]
    By the definitions of $\delta_h^t,\delta_h^{t+1}$,
    \[
    \begin{aligned}
        \mathcal E_3
        &\leq
        \eta(\delta_h^t+\delta_h^{t+1})
        \sum_{i=1}^{N}
        \left\|
            \bar\pi_{h,i}^{t+1}
            -
            \pi_{h,i}^{*,\tau}
        \right\|_1
        \\
        &\leq
        \eta(\delta_h^t+\delta_h^{t+1})
        \sqrt{
            2N
            \KL(
                \pi_h^{*,\tau}
                ||
                \bar\pi_h^{t+1}
            )
        }
        \\
        &\leq
        \eta N^2H
        \KL(
            \pi_h^{*,\tau}
            ||
            \bar\pi_h^{t+1}
        )
        +
        4\eta N^2
        (\delta_h^t+\delta_h^{t+1})^2.
    \end{aligned}
    \]
    The $Q_h^{*,\tau}$ part is bounded similarly to inequality (\ref{tmpeq2}):
    \[
    \begin{aligned}
        &
        \sum_{i=1}^{N}
        \left\langle
            \bar\pi_{h,i}^{t+1}
            -
            \pi_{h,i}^{*,\tau},
            Q_{h,i}^{*,\tau}
            (s,\cdot,\bar\pi_{h,-i}^{t})
            -
            Q_{h,i}^{*,\tau}
            (s,\cdot,\bar\pi_{h,-i}^{t+1})
        \right\rangle
        \\
        &\leq
        N^2H
        \KL(\pi_h^t||\bar\pi_h^t)(s)
        +
        N^2H
        \KL(\bar\pi_h^{t+1}||\pi_h^t)(s)
        \\
        &\quad+
        2N^2H
        \KL(\pi_h^{*,\tau}||\bar\pi_h^{t+1})(s).
    \end{aligned}
    \]
    Therefore,
    \[
    \begin{aligned}
        \mathcal J
        &\leq
        \eta N^2H
        \KL(\pi_h^t||\bar\pi_h^t)(s)
        +
        \eta N^2H
        \KL(\bar\pi_h^{t+1}||\pi_h^t)(s)
        \\
        &\quad+
        3\eta N^2H
        \KL(\pi_h^{*,\tau}||\bar\pi_h^{t+1})(s)
        +
        4\eta N^2
        (\delta_h^t+\delta_h^{t+1})^2.
    \end{aligned}
    \]
    Plugging this into the KL identity gives
    \[
    \begin{aligned}
        &
        \KL(\pi_h^{*,\tau}||\bar\pi_h^{t+1})(s)
        \\
        &\leq
        \KL(\pi_h^{*,\tau}||\pi_h^{t+1})(s)
        -
        \KL(\bar\pi_h^{t+1}||\pi_h^{t+1})(s)
        \\
        &\quad+
        \eta N^2H
        \KL(\pi_h^t||\bar\pi_h^t)(s)
        +
        \eta N^2H
        \KL(\bar\pi_h^{t+1}||\pi_h^t)(s)
        \\
        &\quad+
        3\eta N^2H
        \KL(\pi_h^{*,\tau}||\bar\pi_h^{t+1})(s)
        +
        4\eta N^2
        (\delta_h^t+\delta_h^{t+1})^2.
    \end{aligned}
    \]
    Moving the term involving
    $\KL(\pi_h^{*,\tau}||\bar\pi_h^{t+1})$
    to the left-hand side yields
    \[
    \begin{aligned}
        &
        (1-3\eta N^2H)
        \KL(\pi_h^{*,\tau}||\bar\pi_h^{t+1})(s)
        \\
        &\leq
        \KL(\pi_h^{*,\tau}||\pi_h^{t+1})(s)
        -
        \KL(\bar\pi_h^{t+1}||\pi_h^{t+1})(s)
        \\
        &\quad+
        \eta N^2H
        \KL(\pi_h^t||\bar\pi_h^t)(s)
        +
        \eta N^2H
        \KL(\bar\pi_h^{t+1}||\pi_h^t)(s)
        \\
        &\quad+
        4\eta N^2
        (\delta_h^t+\delta_h^{t+1})^2.
    \end{aligned}
    \]
    Now substituting the bound on
    $\KL(\pi_h^{*,\tau}||\pi_h^{t+1})(s)$
    from inequality (\ref{tmpeq}), we obtain
    \[
    \begin{aligned}
        &
        (1-3\eta N^2H)
        \KL(\pi_h^{*,\tau}||\bar\pi_h^{t+1})(s)
        \\
        &\leq
        \beta
        \KL(\pi_h^{*,\tau}||\pi_h^t)(s)
        +
        \eta N^2H
        \KL(\pi_h^t||\bar\pi_h^t)(s)
        \\
        &\quad
        -
        (\beta-\eta N^2H)
        \KL(\bar\pi_h^{t+1}||\pi_h^t)(s)
        \\
        &\quad
        -
        (1-3\eta N^2H)
        \KL(\pi_h^{t+1}||\bar\pi_h^{t+1})(s)
        \\
        &\quad
        -
        \frac{3\eta\tau}{4}
        \KL(\bar\pi_h^{t+1}||\pi_h^{*,\tau})(s)
        \\
        &\quad+
        \frac{2\eta N}{\tau}
        (\delta_h^{t+1})^2
        +
        4\eta N^2
        (\delta_h^t+\delta_h^{t+1})^2
        \\
        &\quad-
        \KL(\bar\pi_h^{t+1}||\pi_h^{t+1})(s)
        \\
        &\quad+
        \eta N^2H
        \KL(\pi_h^t||\bar\pi_h^t)(s)
        +
        \eta N^2H
        \KL(\bar\pi_h^{t+1}||\pi_h^t)(s)
        \\
        &\quad+
        4\eta N^2
        (\delta_h^t+\delta_h^{t+1})^2.\\
    \end{aligned}
    \]
    Then we combine the second line from the bottom into the previously occurred terms and discard all terms with negative coefficients.
    \[
    \begin{aligned}
        &
        (1-3\eta N^2H)
        \KL(\pi_h^{*,\tau}||\bar\pi_h^{t+1})(s)
        \\
        &\leq
        \beta
        \KL(\pi_h^{*,\tau}||\pi_h^t)(s)
        +
        2\eta N^2H
        \KL(\pi_h^t||\bar\pi_h^t)(s)
        \\
        &\quad
        -
        (\beta-2\eta N^2H)
        \KL(\bar\pi_h^{t+1}||\pi_h^t)(s)
        \\
        &\quad
        -
        (1-3\eta N^2H)
        \KL(\pi_h^{t+1}||\bar\pi_h^{t+1})(s)
        \\
        &\quad
        -
        \frac{3\eta\tau}{4}
        \KL(\bar\pi_h^{t+1}||\pi_h^{*,\tau})(s)
        \\
        &\quad+
        \frac{2\eta N}{\tau}
        (\delta_h^{t+1})^2
        +
        4\eta N^2
        (\delta_h^t+\delta_h^{t+1})^2
        \\
        &\quad-
        \KL(\bar\pi_h^{t+1}||\pi_h^{t+1})(s)
        \\
        &\quad+
        4\eta N^2
        (\delta_h^t+\delta_h^{t+1})^2\\
        &\leq
        \beta
        \KL(\pi_h^{*,\tau}||\pi_h^t)(s)
        +
        2\eta N^2H
        \KL(\pi_h^t||\bar\pi_h^t)(s)
        \\
        &\quad+
        \frac{2\eta N}{\tau}
        (\delta_h^{t+1})^2
        +
        4\eta N^2
        (\delta_h^t+\delta_h^{t+1})^2
        \\
        &\quad+
        4\eta N^2
        (\delta_h^t+\delta_h^{t+1})^2\\
    \end{aligned}
    \]
    In particular,
    \[
        \beta-2\eta N^2H\geq0,
    \]
    because
    \[
    \beta\geq \frac{1}{2}\geq \frac{2}{16}\geq 2\eta N^2H.
    \]
    under the step-size condition. Therefore, eventually we get
    \[
    \begin{aligned}
        &
        (1-3\eta N^2H)
        \KL(\pi_h^{*,\tau}||\bar\pi_h^{t+1})(s)
        \\
        &\leq
        \beta
        \KL(\pi_h^{*,\tau}||\pi_h^t)(s)
        +
        2\eta N^2H
        \KL(\pi_h^t||\bar\pi_h^t)(s)
        \\
        &\quad+
        \frac{2\eta N}{\tau}
        (\delta_h^{t+1})^2
        +
        8\eta N^2
        (\delta_h^t+\delta_h^{t+1})^2.
    \end{aligned}
    \]
    Since $\eta N^2H\leq1/16$, we have
    \[
        1-3\eta N^2H\geq\frac12.
    \]
    Hence
    \[
    \begin{aligned}
        &
        \KL(\pi_h^{*,\tau}||\bar\pi_h^{t+1})(s)
        \\
        &\leq
        2\beta
        \KL(\pi_h^{*,\tau}||\pi_h^t)(s)
        +
        4\eta N^2H
        \KL(\pi_h^t||\bar\pi_h^t)(s)
        \\
        &\quad+
        \frac{4\eta N}{\tau}
        (\delta_h^{t+1})^2
        +
        16\eta N^2
        (\delta_h^t+\delta_h^{t+1})^2.
    \end{aligned}
    \]
    Since a simple calculation shows that
    \[
        2\beta\leq4,
        \qquad
        4\eta N^2H\leq \frac14\leq 3\leq 4(1-4\eta N^2H)=4\lambda_\eta,
    \]
    and the quadratic terms are bounded by
    $\Gamma_\eta(\Delta_h^t)^2$, we obtain
    \[
        \KL(\pi_h^{*,\tau}||\bar\pi_h^{t+1})(s)
        \leq
        4
        \left[
            \KL(\pi_h^{*,\tau}||\pi_h^t)(s)
            +
            \lambda_\eta
            \KL(\pi_h^t||\bar\pi_h^t)(s)
        \right]
        +
        \Gamma_\eta
        (\Delta_h^t)^2.
    \]
    Taking the maximum over $s\in\gS$ proves
    \[
        \bar L_h^{t+1}
        \leq
        4L_h^t
        +
        \Gamma_\eta
        \left(
            \delta_h^t+2\delta_h^{t+1}
        \right)^2.
    \]

    \paragraph{Step 7: passing to square-root variables.}
    Recall that
    \[
    \omega_\eta=\sqrt{\Gamma_\eta}.
    \]
    By the definition
    \[
        X_h^t=\sqrt{L_h^t},
        \qquad
        \bar X_h^t=\sqrt{\bar L_h^t},
    \]
    and by the elementary inequality
    \[
        \sqrt{a+b}\leq \sqrt a+\sqrt b
        \qquad
        \text{for }a,b\geq0,
    \]
    the first recursion implies
    \[
    \boxed
    {\begin{aligned}
        X_h^{t+1}
        &=
        \sqrt{L_h^{t+1}}
        \\
        &\leq
        \sqrt{\beta L_h^t}
        +
        \sqrt{\Gamma_\eta}
        \left(
            \delta_h^t+2\delta_h^{t+1}
        \right)
        \\
        &=
        \alpha X_h^t
        +
        \omega_\eta
        \left(
            \delta_h^t+2\delta_h^{t+1}
        \right).
    \end{aligned}}
    \]
    Similarly, the second recursion implies
    \[
    \boxed
    {\begin{aligned}
        \bar X_h^{t+1}
        &=
        \sqrt{\bar L_h^{t+1}}
        \\
        &\leq
        2X_h^t
        +
        \omega_\eta
        \left(
            \delta_h^t+2\delta_h^{t+1}
        \right).
    \end{aligned}}
    \]
    This completes the proof.
\end{proof}

\subsection{Proof of Lemma~\ref{iter2}}

\begin{proof}
    We prove the recursion for a fixed $h\in[H]$ and $t\geq0$. Recall that
    \[
        \beta:=1-\eta\tau.
    \]
    The case $t=0$ is immediate from the initialization and the boundedness of the finite-horizon value functions:
    \[
        \delta_h^0
        =
        \max_{s,i,\va}
        \left|
            Q_{h,i}^0(s,\va)-Q_{h,i}^{*,\tau}(s,\va)
        \right|
        \leq NH
        =
        NH\beta^0,
    \]
    while the summation on the right-hand side is empty. Hence we focus on $t\geq1$.

    We first estimate the value-function error at step $h+1$. If $h=H$, then $V_{H+1,i}^t\equiv V_{H+1,i}^{*,\tau}\equiv0$, and the desired estimate holds under the convention
    \[
        \delta_{H+1}^t=0,
        \qquad
        \bar X_{H+1}^t=0.
    \]
    Hence assume $h+1\in[H]$. For every $s'\in\gS$,
    \[
        V_{h+1,i}^{t}(s')
        =
        Q_{h+1,i}^{t}
        (s',\bar\pi_{h+1}^{t}(s')),
        \qquad
        V_{h+1,i}^{*,\tau}(s')
        =
        Q_{h+1,i}^{*,\tau}
        (s',\pi_{h+1}^{*,\tau}(s')).
    \]
    Therefore,
    \[
    \begin{aligned}
        &
        \left|
            V_{h+1,i}^{t}(s')
            -
            V_{h+1,i}^{*,\tau}(s')
        \right|
        \\
        &\leq
        \left|
            Q_{h+1,i}^{t}
            (s',\bar\pi_{h+1}^{t}(s'))
            -
            Q_{h+1,i}^{*,\tau}
            (s',\bar\pi_{h+1}^{t}(s'))
        \right|
        \\
        &\quad+
        \left|
            Q_{h+1,i}^{*,\tau}
            (s',\bar\pi_{h+1}^{t}(s'))
            -
            Q_{h+1,i}^{*,\tau}
            (s',\pi_{h+1}^{*,\tau}(s'))
        \right|.
    \end{aligned}
    \]
    The first term is at most $\delta_{h+1}^t$. For the second term, Lemma~\ref{lem:multilinear-tv} with the uniform bound $NH$ gives
    \[
    \begin{aligned}
        &
        \left|
            Q_{h+1,i}^{*,\tau}
            (s',\bar\pi_{h+1}^{t}(s'))
            -
            Q_{h+1,i}^{*,\tau}
            (s',\pi_{h+1}^{*,\tau}(s'))
        \right|
        \\
        &\leq
        NH
        \sum_{m=1}^{N}
        \left\|
            \bar\pi_{h+1,m}^{t}(\cdot|s')
            -
            \pi_{h+1,m}^{*,\tau}(\cdot|s')
        \right\|_1.
    \end{aligned}
    \]
    By Cauchy--Schwarz and Pinsker's inequality in Lemma~\ref{lem:pinsker},
    \[
    \begin{aligned}
        &
        \sum_{m=1}^{N}
        \left\|
            \bar\pi_{h+1,m}^{t}(\cdot|s')
            -
            \pi_{h+1,m}^{*,\tau}(\cdot|s')
        \right\|_1
        \\
        &\leq
        \sqrt{
            2N
            \KL\left(
                \pi_{h+1}^{*,\tau}(s')
                ||
                \bar\pi_{h+1}^{t}(s')
            \right)
        }
        \leq
        \sqrt{2N}\,\bar X_{h+1}^t.
    \end{aligned}
    \]
    Combining the preceding estimates and taking the maximum over $s'\in\gS$ gives
    \begin{equation}
        \label{eq:value-error-for-q-tracking}
        \max_{s'\in\gS}
        \left|
            V_{h+1,i}^{t}(s')
            -
            V_{h+1,i}^{*,\tau}(s')
        \right|
        \leq
        \delta_{h+1}^t
        +
        NH\sqrt{2N}\,\bar X_{h+1}^t.
    \end{equation}

    We now turn to the $Q$-function recursion. Summing the separated update rule over $j\in\gN_i$ gives, for every $s\in\gS$ and $\va\in\gA$,
    \[
    \begin{aligned}
        Q_{h,i}^{t+1}(s,\va)
        &=
        \sum_{j\in\gN_i}
        Q_{h,i,j}^{t+1}(s,a_i,a_j)
        \\
        &=
        \beta
        \sum_{j\in\gN_i}
        Q_{h,i,j}^{t}(s,a_i,a_j)
        \\
        &\quad+
        \eta\tau
        \sum_{j\in\gN_i}
        \left[
            r_{h,i,j}(s,a_i,a_j)
            +
            \sK_{h,i,j}(\cdot|s,a_i,a_j)
            V_{h+1,i}^{t+1}(\cdot)
        \right]
        \\
        &=
        \beta Q_{h,i}^{t}(s,\va)
        +
        \eta\tau
        \left[
            r_{h,i}(s,\va)
            +
            \sP_h(\cdot|s,\va)V_{h+1,i}^{t+1}(\cdot)
        \right].
    \end{aligned}
    \]
    Here we used the definitions
    \[
        Q_{h,i}^{t}
        =
        \sum_{j\in\gN_i}Q_{h,i,j}^{t},
        \qquad
        r_{h,i}
        =
        \sum_{j\in\gN_i}r_{h,i,j},
    \]
    and the transition identity
    \[
        \sP_h(\cdot|s,\va)
        =
        \sum_{j\in\gN_i}
        \sK_{h,i,j}(\cdot|s,a_i,a_j).
    \]
    This identity holds in both controller cases. In particular, if
    $\gN_C=\emptyset$, each summand is
    $\sP_{h,o}(\cdot|s)/|\gN_i|$, so summing over the neighbors restores
    the entire action-independent continuation term. Thus the aggregated
    Bellman recursion above, and every estimate below, apply without
    assuming a nonempty controller set.

    We claim that for every $t\geq0$,
    \begin{equation}
        \label{eq:explicit-Qt}
        Q_{h,i}^{t}(s,\va)
        =
        r_{h,i}(s,\va)
        +
        \eta\tau
        \sum_{k=0}^{t-1}
        \beta^k
        \sP_h(\cdot|s,\va)
        V_{h+1,i}^{t-k}(\cdot).
    \end{equation}
    When $t=0$, the summation is empty and the identity follows from the initialization
    \[
        Q_{h,i}^0=r_{h,i}.
    \]
    Suppose the identity holds for $t$. Then
    \[
    \begin{aligned}
        Q_{h,i}^{t+1}(s,\va)
        &=
        \beta Q_{h,i}^{t}(s,\va)
        +
        \eta\tau
        \left[
            r_{h,i}(s,\va)
            +
            \sP_h(\cdot|s,\va)V_{h+1,i}^{t+1}(\cdot)
        \right]
        \\
        &=
        \beta
        \left[
            r_{h,i}(s,\va)
            +
            \eta\tau
            \sum_{k=0}^{t-1}
            \beta^k
            \sP_h(\cdot|s,\va)V_{h+1,i}^{t-k}(\cdot)
        \right]
        \\
        &\quad+
        \eta\tau r_{h,i}(s,\va)
        +
        \eta\tau
        \sP_h(\cdot|s,\va)V_{h+1,i}^{t+1}(\cdot)
        \\
        &=
        (\beta+\eta\tau)r_{h,i}(s,\va)
        +
        \eta\tau
        \sP_h(\cdot|s,\va)V_{h+1,i}^{t+1}(\cdot)
        \\
        &\quad+
        \eta\tau
        \sum_{k=1}^{t}
        \beta^k
        \sP_h(\cdot|s,\va)V_{h+1,i}^{t+1-k}(\cdot)
        \\
        &=
        r_{h,i}(s,\va)
        +
        \eta\tau
        \sum_{k=0}^{t}
        \beta^k
        \sP_h(\cdot|s,\va)V_{h+1,i}^{t+1-k}(\cdot),
    \end{aligned}
    \]
    because $\beta+\eta\tau=1$. This proves \eqref{eq:explicit-Qt} by induction.

    On the other hand,
    \[
        Q_{h,i}^{*,\tau}(s,\va)
        =
        r_{h,i}(s,\va)
        +
        \sP_h(\cdot|s,\va)
        V_{h+1,i}^{*,\tau}(\cdot).
    \]
    Since
    \[
        \eta\tau
        \sum_{k=0}^{t-1}\beta^k
        =
        1-\beta^t,
    \]
    we can rewrite this as
    \begin{equation}
        \label{eq:explicit-Qstar}
        \begin{aligned}
        Q_{h,i}^{*,\tau}(s,\va)
        &=
        r_{h,i}(s,\va)
        +
        \eta\tau
        \sum_{k=0}^{t-1}
        \beta^k
        \sP_h(\cdot|s,\va)
        V_{h+1,i}^{*,\tau}(\cdot)
        \\
        &\quad+
        \beta^t
        \sP_h(\cdot|s,\va)
        V_{h+1,i}^{*,\tau}(\cdot).
        \end{aligned}
    \end{equation}
    Subtracting \eqref{eq:explicit-Qstar} from \eqref{eq:explicit-Qt}, we obtain
    \[
    \begin{aligned}
        &
        \left|
            Q_{h,i}^{t}(s,\va)
            -
            Q_{h,i}^{*,\tau}(s,\va)
        \right|
        \\
        &=
        \Bigg|
            \eta\tau
            \sum_{k=0}^{t-1}
            \beta^k
            \sP_h(\cdot|s,\va)
            \left(
                V_{h+1,i}^{t-k}(\cdot)
                -
                V_{h+1,i}^{*,\tau}(\cdot)
            \right)
            \\
        &\quad
            -
            \beta^t
            \sP_h(\cdot|s,\va)
            V_{h+1,i}^{*,\tau}(\cdot)
        \Bigg|
        \\
        &\leq
        \eta\tau
        \sum_{k=0}^{t-1}
        \beta^k
        \max_{s'\in\gS}
        \left|
            V_{h+1,i}^{t-k}(s')
            -
            V_{h+1,i}^{*,\tau}(s')
        \right|
        \\
        &\quad+
        \beta^t
        \left\|
            V_{h+1,i}^{*,\tau}
        \right\|_\infty.
    \end{aligned}
    \]
    Since the value functions are bounded by $NH$,
    \[
        \left\|
            V_{h+1,i}^{*,\tau}
        \right\|_\infty
        \leq NH.
    \]
    Therefore,
    \[
    \begin{aligned}
        \left|
            Q_{h,i}^{t}(s,\va)
            -
            Q_{h,i}^{*,\tau}(s,\va)
        \right|
        \leq
        \eta\tau
        \sum_{k=0}^{t-1}
        \beta^k
        \max_{s'\in\gS}
        \left|
            V_{h+1,i}^{t-k}(s')
            -
            V_{h+1,i}^{*,\tau}(s')
        \right|
        +
        NH\beta^t.
    \end{aligned}
    \]
    Applying the value-error estimate \eqref{eq:value-error-for-q-tracking} with time index $t-k$ gives
    \[
    \begin{aligned}
        \left|
            Q_{h,i}^{t}(s,\va)
            -
            Q_{h,i}^{*,\tau}(s,\va)
        \right|
        \leq
        \eta\tau
        \sum_{k=0}^{t-1}
        \beta^k
        \left(
            \delta_{h+1}^{t-k}
            +
            NH\sqrt{2N}\,
            \bar X_{h+1}^{t-k}
        \right)
        +
        NH\beta^t.
    \end{aligned}
    \]
    Finally, taking the maximum over $s\in\gS$, $i\in[N]$, and $\va\in\gA$ yields
    \[
        \delta_h^t
        \leq
        \eta\tau
        \sum_{k=0}^{t-1}
        \beta^k
        \left(
            \delta_{h+1}^{t-k}
            +
            NH\sqrt{2N}\,
            \bar X_{h+1}^{t-k}
        \right)
        +
        NH\beta^t.
    \]
    This proves the lemma.
\end{proof}

\section{Proof of Lemma~\ref{gregapest}}
\label{app:qre-gap-proof}

\begin{proof}
    Fix $h\in[H]$ and $t\geq1$. Write $q_h:=H-h$ and
    $\alpha:=\sqrt{1-\eta\tau}$. Throughout the proof, for a fixed state $s\in\gS$, we write
    \[
        u_i(\rho)
        :=
        Q_{h,i}^{*,\tau}(s,\rho),
        \qquad
        \rho=(\rho_1,\ldots,\rho_N)\in\prod_{i=1}^{N}\Delta(\gA_i).
    \]
    We also use the shorthand
    \[
        \pi^*
        :=
        \pi_h^{*,\tau}(s),
        \qquad
        \bar\pi
        :=
        \bar\pi_h^t(s).
    \]
    For a unilateral deviation $\nu_i\in\Delta(\gA_i)$, define the entropy-regularized deviation advantage
    \[
    \begin{aligned}
        \mathsf{Adv}_i(\nu_i;\bar\pi)
        :=
        \left\langle
            \nu_i-\bar\pi_i,
            Q_{h,i}^{*,\tau}(s,\cdot,\bar\pi_{-i})
        \right\rangle
        +
        \tau\gH(\nu_i)
        -
        \tau\gH(\bar\pi_i).
    \end{aligned}
    \]
    By definition,
    \[
        \textnormal{QRE-Gap}
        (Q_h^{*,\tau}(s),\bar\pi_h^t(s))
        =
        \max_{i\in[N]}
        \max_{\nu_i\in\Delta(\gA_i)}
        \mathsf{Adv}_i(\nu_i;\bar\pi).
    \]
    Since choosing $\nu_i=\bar\pi_i$ gives zero advantage, each maximized advantage is nonnegative. Hence
    \[
    \begin{aligned}
        \textnormal{QRE-Gap}
        (Q_h^{*,\tau}(s),\bar\pi_h^t(s))
        &\leq
        \sum_{i=1}^{N}
        \max_{\nu_i\in\Delta(\gA_i)}
        \mathsf{Adv}_i(\nu_i;\bar\pi).
    \end{aligned}
    \]
    Let $\nu_i^\dagger\in\Delta(\gA_i)$ be a maximizer of the $i$-th term, and denote
    \[
        \nu^\dagger
        :=
        (\nu_1^\dagger,\ldots,\nu_N^\dagger).
    \]
    It is therefore sufficient to upper-bound
    \[
        \sum_{i=1}^{N}
        \mathsf{Adv}_i(\nu_i^\dagger;\bar\pi).
    \]

    We first expand this quantity. By Lemma~\ref{lem:qre-induced-structure}, the stage game
    $Q_h^{*,\tau}(s)$ is zero-sum networked separable. Hence, for every product profile $\rho$,
    \[
        \sum_{i=1}^{N}u_i(\rho)=0.
    \]
    In particular,
    \[
        \sum_{i=1}^{N}
        u_i(\bar\pi)=0,
        \qquad
        \sum_{i=1}^{N}
        u_i(\pi^*)=0.
    \]
    Thus
    \[
    \begin{aligned}
        \sum_{i=1}^{N}
        \mathsf{Adv}_i(\nu_i^\dagger;\bar\pi)
        &=
        \sum_{i=1}^{N}
        \left[
            u_i(\nu_i^\dagger,\bar\pi_{-i})
            -
            u_i(\bar\pi)
            +
            \tau\gH(\nu_i^\dagger)
            -
            \tau\gH(\bar\pi_i)
        \right]
        \\
        &=
        \sum_{i=1}^{N}
        \left[
            u_i(\nu_i^\dagger,\bar\pi_{-i})
            -
            u_i(\pi^*)
            +
            \tau\gH(\nu_i^\dagger)
            -
            \tau\gH(\bar\pi_i)
        \right].
    \end{aligned}
    \]
    We decompose the last display into three terms:
    \[
    \begin{aligned}
        \sum_{i=1}^{N}
        \mathsf{Adv}_i(\nu_i^\dagger;\bar\pi)
        =
        \mathrm{I}
        +
        \mathrm{II}
        +
        \mathrm{III},
    \end{aligned}
    \]
    where
    \[
    \begin{aligned}
        \mathrm{I}
        &:=
        \sum_{i=1}^{N}
        \Big[
            u_i(\nu_i^\dagger,\bar\pi_{-i})
            -
            u_i(\nu_i^\dagger,\pi_{-i}^*)
            -
            u_i(\pi_i^*,\bar\pi_{-i})
            +
            u_i(\pi^*)
        \Big],
    \end{aligned}
    \]
    \[
    \begin{aligned}
        \mathrm{II}
        &:=
        \sum_{i=1}^{N}
        \Big[
            u_i(\pi_i^*,\bar\pi_{-i})
            -
            u_i(\pi^*)
            +
            \tau\gH(\pi_i^*)
            -
            \tau\gH(\bar\pi_i)
        \Big],
    \end{aligned}
    \]
    and
    \[
    \begin{aligned}
        \mathrm{III}
        &:=
        \sum_{i=1}^{N}
        \Big[
            u_i(\nu_i^\dagger,\pi_{-i}^*)
            -
            u_i(\pi^*)
            +
            \tau\gH(\nu_i^\dagger)
            -
            \tau\gH(\pi_i^*)
        \Big].
    \end{aligned}
    \]

    We now control these three terms separately.

    For the first term, using the networked separable representation of $Q_h^{*,\tau}(s)$, we have
    \[
    \begin{aligned}
        \mathrm{I}
        &=
        \sum_{i=1}^{N}
        Q_{h,i}^{*,\tau}
        \left(
            s,
            \nu_i^\dagger-\pi_i^*,
            \bar\pi_{-i}-\pi_{-i}^*
        \right)
        \\
        &=
        \sum_{i=1}^{N}
        \sum_{j\in\gN_i}
        \left(
            \nu_i^\dagger-\pi_i^*
        \right)^\top
        Q_{h,i,j}^{*,\tau}(s,\cdot,\cdot)
        \left(
            \bar\pi_j-\pi_j^*
        \right).
    \end{aligned}
    \]
    By the edge bound in Eq.~(\ref{eq:edge-q-bound}), valid for both controller cases,
    \[
        \left\|
            Q_{h,i,j}^{*,\tau}
        \right\|_\infty
        \leq NH.
    \]
    It follows that
    \[
    \begin{aligned}
        \mathrm{I}
        &\leq
        NH
        \sum_{i=1}^{N}
        \sum_{j\in\gN_i}
        \left\|
            \nu_i^\dagger-\pi_i^*
        \right\|_1
        \left\|
            \bar\pi_j-\pi_j^*
        \right\|_1
        \\
        &\leq
        NH
        \sum_{i=1}^{N}
        \sum_{j=1}^{N}
        \left\|
            \nu_i^\dagger-\pi_i^*
        \right\|_1
        \left\|
            \bar\pi_j-\pi_j^*
        \right\|_1.
    \end{aligned}
    \]
    Applying Young's inequality with the weights
    \[
        a b
        \leq
        \frac12
        \left(
            \frac{\tau}{N^2H}a^2
            +
            \frac{N^2H}{\tau}b^2
        \right),
    \]
    we obtain
    \[
    \begin{aligned}
        \mathrm{I}
        &\leq
        NH
        \sum_{i=1}^{N}
        \sum_{j=1}^{N}
        \frac12
        \left[
            \frac{\tau}{N^2H}
            \left\|
                \nu_i^\dagger-\pi_i^*
            \right\|_1^2
            +
            \frac{N^2H}{\tau}
            \left\|
                \bar\pi_j-\pi_j^*
            \right\|_1^2
        \right]
        \\
        &=
        \frac{\tau}{2N}
        \sum_{i=1}^{N}
        \sum_{j=1}^{N}
        \left\|
            \nu_i^\dagger-\pi_i^*
        \right\|_1^2
        +
        \frac{N^3H^2}{2\tau}
        \sum_{i=1}^{N}
        \sum_{j=1}^{N}
        \left\|
            \bar\pi_j-\pi_j^*
        \right\|_1^2
        \\
        &=
        \frac{\tau}{2}
        \sum_{i=1}^{N}
        \left\|
            \nu_i^\dagger-\pi_i^*
        \right\|_1^2
        +
        \frac{N^4H^2}{2\tau}
        \sum_{j=1}^{N}
        \left\|
            \bar\pi_j-\pi_j^*
        \right\|_1^2.
    \end{aligned}
    \]
    By Pinsker's inequality in Lemma~\ref{lem:pinsker},
    \[
        \frac12
        \left\|
            \nu_i^\dagger-\pi_i^*
        \right\|_1^2
        \leq
        \KL(\nu_i^\dagger||\pi_i^*),
    \]
    and
    \[
        \frac12
        \left\|
            \bar\pi_j-\pi_j^*
        \right\|_1^2
        \leq
        \KL(\pi_j^*||\bar\pi_j).
    \]
    Therefore,
    \[
    \begin{aligned}
        \mathrm{I}
        &\leq
        \tau
        \sum_{i=1}^{N}
        \KL(\nu_i^\dagger||\pi_i^*)
        +
        \frac{N^4H^2}{\tau}
        \sum_{j=1}^{N}
        \KL(\pi_j^*||\bar\pi_j)
        \\
        &=
        \tau
        \KL(\nu^\dagger||\pi^*)
        +
        \frac{N^4H^2}{\tau}
        \KL(\pi^*||\bar\pi).
    \end{aligned}
    \]

    We next analyze $\mathrm{II}$. By Lemma~\ref{lem:qre-induced-structure}, $Q_h^{*,\tau}(s)$ is zero-sum networked separable. Applying Lemma~\ref{sumzero2} to the two product profiles $\pi^*$ and $\bar\pi$ gives
    \[
        \sum_{i=1}^{N}
        u_i(\pi_i^*,\bar\pi_{-i})
        =
        -
        \sum_{i=1}^{N}
        u_i(\bar\pi_i,\pi_{-i}^*).
    \]
    Moreover, the QRE first-order condition gives
    \[
        \tau\log \pi_i^*
        \overset{\boldsymbol{1}}{=}
        Q_{h,i}^{*,\tau}(s,\cdot,\pi_{-i}^*),
    \]
    so that
    \[
    \begin{aligned}
        u_i(\bar\pi_i,\pi_{-i}^*)
        -
        u_i(\pi^*)
        =
        \left\langle
            \bar\pi_i-\pi_i^*,
            Q_{h,i}^{*,\tau}(s,\cdot,\pi_{-i}^*)
        \right\rangle
        =
        \tau
        \left\langle
            \bar\pi_i-\pi_i^*,
            \log \pi_i^*
        \right\rangle.
    \end{aligned}
    \]
    Since $\sum_i u_i(\pi^*)=0$, we get
    \[
    \begin{aligned}
        \sum_{i=1}^{N}
        \left[
            u_i(\pi_i^*,\bar\pi_{-i})
            -
            u_i(\pi^*)
        \right]
        &=
        \sum_{i=1}^{N}
        u_i(\pi_i^*,\bar\pi_{-i})
        \\
        &=
        -
        \sum_{i=1}^{N}
        u_i(\bar\pi_i,\pi_{-i}^*)
        \\
        &=
        -
        \sum_{i=1}^{N}
        \left[
            u_i(\pi^*)
            +
            \tau
            \left\langle
                \bar\pi_i-\pi_i^*,
                \log\pi_i^*
            \right\rangle
        \right]
        \\
        &=
        -
        \tau
        \sum_{i=1}^{N}
        \left\langle
            \bar\pi_i-\pi_i^*,
            \log\pi_i^*
        \right\rangle.
    \end{aligned}
    \]
    Therefore,
    \[
    \begin{aligned}
        \mathrm{II}
        &=
        -
        \tau
        \sum_{i=1}^{N}
        \left\langle
            \bar\pi_i-\pi_i^*,
            \log\pi_i^*
        \right\rangle
        +
        \tau
        \sum_{i=1}^{N}
        \left[
            \gH(\pi_i^*)
            -
            \gH(\bar\pi_i)
        \right]
        \\
        &=
        -
        \tau
        \sum_{i=1}^{N}
        \left\langle
            \bar\pi_i-\pi_i^*,
            \log\pi_i^*
        \right\rangle
        -
        \tau
        \sum_{i=1}^{N}
        \left\langle
            \pi_i^*,
            \log\pi_i^*
        \right\rangle
        +
        \tau
        \sum_{i=1}^{N}
        \left\langle
            \bar\pi_i,
            \log\bar\pi_i
        \right\rangle
        \\
        &=
        \tau
        \sum_{i=1}^{N}
        \left\langle
            \bar\pi_i,
            \log\bar\pi_i-\log\pi_i^*
        \right\rangle
        \\
        &=
        \tau
        \KL(\bar\pi||\pi^*).
    \end{aligned}
    \]

    The third term $\mathrm{III}$ is handled similarly. By the QRE first-order condition,
    \[
    \begin{aligned}
        u_i(\nu_i^\dagger,\pi_{-i}^*)
        -
        u_i(\pi^*)
        &=
        \left\langle
            \nu_i^\dagger-\pi_i^*,
            Q_{h,i}^{*,\tau}(s,\cdot,\pi_{-i}^*)
        \right\rangle
        \\
        &=
        \tau
        \left\langle
            \nu_i^\dagger-\pi_i^*,
            \log\pi_i^*
        \right\rangle.
    \end{aligned}
    \]
    Thus
    \[
    \begin{aligned}
        \mathrm{III}
        &=
        \tau
        \sum_{i=1}^{N}
        \left\langle
            \nu_i^\dagger-\pi_i^*,
            \log\pi_i^*
        \right\rangle
        +
        \tau
        \sum_{i=1}^{N}
        \left[
            \gH(\nu_i^\dagger)
            -
            \gH(\pi_i^*)
        \right]
        \\
        &=
        \tau
        \sum_{i=1}^{N}
        \left\langle
            \nu_i^\dagger-\pi_i^*,
            \log\pi_i^*
        \right\rangle
        -
        \tau
        \sum_{i=1}^{N}
        \left\langle
            \nu_i^\dagger,
            \log\nu_i^\dagger
        \right\rangle
        +
        \tau
        \sum_{i=1}^{N}
        \left\langle
            \pi_i^*,
            \log\pi_i^*
        \right\rangle
        \\
        &=
        -
        \tau
        \sum_{i=1}^{N}
        \left\langle
            \nu_i^\dagger,
            \log\nu_i^\dagger-\log\pi_i^*
        \right\rangle
        \\
        &=
        -
        \tau
        \KL(\nu^\dagger||\pi^*).
    \end{aligned}
    \]

    Combining the estimates for $\mathrm{I}$, $\mathrm{II}$, and $\mathrm{III}$, we obtain
    \[
    \begin{aligned}
        \sum_{i=1}^{N}
        \mathsf{Adv}_i(\nu_i^\dagger;\bar\pi)
        &=
        \mathrm{I}
        +
        \mathrm{II}
        +
        \mathrm{III}
        \\
        &\leq
        \tau\KL(\nu^\dagger||\pi^*)
        +
        \frac{N^4H^2}{\tau}
        \KL(\pi^*||\bar\pi)
        +
        \tau\KL(\bar\pi||\pi^*)
        -
        \tau\KL(\nu^\dagger||\pi^*)
        \\
        &=
        \tau\KL(\bar\pi||\pi^*)
        +
        \frac{N^4H^2}{\tau}
        \KL(\pi^*||\bar\pi).
    \end{aligned}
    \]
    Since $s\in\gS$ was arbitrary, we have shown that
    \[
    \begin{aligned}
        \textnormal{QRE-Gap}_h(\bar\pi_h^t)
        \leq
        \max_{s\in\gS}
        \left[
            \tau
            \KL(
                \bar\pi_h^t(s)
                ||
                \pi_h^{*,\tau}(s)
            )
            +
            \frac{N^4H^2}{\tau}
            \KL(
                \pi_h^{*,\tau}(s)
                ||
                \bar\pi_h^t(s)
            )
        \right].
    \end{aligned}
    \]

    We now bound the two KL terms using the tracking estimates. For the reverse KL term, Lemma~\ref{iter1} gives, for every $s\in\gS$ and $t\geq1$,
    \[
        \frac{\eta\tau}{2}
        \KL(
            \bar\pi_h^t(s)
            ||
            \pi_h^{*,\tau}(s)
        )
        \leq
        \beta L_h^{t-1}
        +
        \Gamma_\eta
        \left(
            \delta_h^{t-1}+2\delta_h^t
        \right)^2.
    \]
    Therefore,
    \[
    \begin{aligned}
        \tau
        \KL(
            \bar\pi_h^t(s)
            ||
            \pi_h^{*,\tau}(s)
        )
        &\leq
        2\eta^{-1}\beta L_h^{t-1}
        +
        2\eta^{-1}\Gamma_\eta
        \left(
            \delta_h^{t-1}+2\delta_h^t
        \right)^2.
    \end{aligned}
    \]
    Since $L_h^{t-1}=(X_h^{t-1})^2$, Lemma~\ref{keygoal2} implies
    \[
    \begin{aligned}
        2\eta^{-1}\beta L_h^{t-1}
        &=
        2\eta^{-1}\alpha^2(X_h^{t-1})^2
        \\
        &\leq
        2\eta^{-1}\alpha^2
        C_h^2
        (t+q_h+1)^{2q_h}
        \alpha^{2t-2}
        \\
        &=
        2\eta^{-1}
        C_h^2
        (t+q_h+1)^{2q_h}
        \alpha^{2t}
        \\
        &\leq
        2\eta^{-1}
        C_h^2
        (t+q_h+2)^{2q_h}
        \alpha^{2t}.
    \end{aligned}
    \]
    Moreover, by Lemma~\ref{keygoal2},
    \[
        \delta_h^{t-1}
        \leq
        C_h(t+q_h+1)^{q_h}\alpha^{t-1}
        \leq
        C_h(t+q_h+2)^{q_h}\alpha^{t-1},
    \]
    and
    \[
        \delta_h^t
        \leq
        C_h(t+q_h+2)^{q_h}\alpha^t
        \leq
        C_h(t+q_h+2)^{q_h}\alpha^{t-1}.
    \]
    Hence
    \[
    \begin{aligned}
        \delta_h^{t-1}+2\delta_h^t
        &\leq
        3C_h(t+q_h+2)^{q_h}\alpha^{t-1}.
    \end{aligned}
    \]
    Since $\eta\tau\leq1/2$, we have
    \[
        \alpha^{-2}
        =
        (1-\eta\tau)^{-1}
        \leq 2.
    \]
    Therefore,
    \[
    \begin{aligned}
        \left(
            \delta_h^{t-1}+2\delta_h^t
        \right)^2
        &\leq
        9C_h^2
        (t+q_h+2)^{2q_h}
        \alpha^{2t-2}
        \\
        &\leq
        18C_h^2
        (t+q_h+2)^{2q_h}
        \alpha^{2t}.
    \end{aligned}
    \]
    Consequently,
    \[
    \begin{aligned}
        \tau
        \KL(
            \bar\pi_h^t(s)
            ||
            \pi_h^{*,\tau}(s)
        )
        &\leq
        \left(
            2\eta^{-1}
            +
            36\eta^{-1}\Gamma_\eta
        \right)
        C_h^2
        (t+q_h+2)^{2q_h}
        \alpha^{2t}.
    \end{aligned}
    \]

    For the forward KL term, Lemma~\ref{keygoal2} gives
    \[
    \begin{aligned}
        \frac{N^4H^2}{\tau}
        \KL(
            \pi_h^{*,\tau}(s)
            ||
            \bar\pi_h^t(s)
        )
        &\leq
        \frac{N^4H^2}{\tau}
        \bar L_h^t
        \\
        &=
        \frac{N^4H^2}{\tau}
        (\bar X_h^t)^2
        \\
        &\leq
        \frac{N^4H^2}{\tau}
        C_h^2
        (t+q_h+2)^{2q_h}
        \alpha^{2t}.
    \end{aligned}
    \]
    Combining the two bounds yields, for every $s\in\gS$,
    \[
    \begin{aligned}
        &
        \tau
        \KL(
            \bar\pi_h^t(s)
            ||
            \pi_h^{*,\tau}(s)
        )
        +
        \frac{N^4H^2}{\tau}
        \KL(
            \pi_h^{*,\tau}(s)
            ||
            \bar\pi_h^t(s)
        )
        \\
        &\leq
        \left(
            2\eta^{-1}
            +
            36\eta^{-1}\Gamma_\eta
            +
            \frac{N^4H^2}{\tau}
        \right)
        C_h^2
        (t+q_h+2)^{2q_h}
        \alpha^{2t}.
    \end{aligned}
    \]
    Since $\alpha^{2t}=(1-\eta\tau)^t$, the preceding estimate gives
    \[
    \begin{aligned}
        \textnormal{QRE-Gap}_h(\bar\pi_h^t)
        &\leq
        \left(
            2\eta^{-1}
            +
            36\eta^{-1}\Gamma_\eta
            +
            \frac{N^4H^2}{\tau}
        \right)
        C_h^2
        (t+q_h+2)^{2q_h}
        \alpha^{2t}
        \\
        &=
        G_h
        (t+q_h+2)^{2q_h}
        \alpha^{2t},
    \end{aligned}
    \]
    where
    \[
        G_h
        :=
        \left(
            \frac{2
            +
            36\Gamma_\eta}{\eta}
            +
            \frac{N^4H^2}{\tau}
        \right)C_h^2.
    \]
    This proves the first displayed estimate. Since
    \[
        \alpha^{2t}=(1-\eta\tau)^t
        \leq
        \exp\left(-\eta\tau t\right),
    \]
    any $t$ satisfying
    \[
        G_h(t+q_h+2)^{2q_h}
        \exp\left(-\eta\tau t\right)
        \leq
        \varepsilon_{\mathrm{qre}}
    \]
    also satisfies
    $\textnormal{QRE-Gap}_h(\bar\pi_h^t)\leq\varepsilon_{\mathrm{qre}}$.
    Solving this inequality gives
    \[
        t
        =
        \widetilde O\left(
            \frac{q_h+\log(G_h/\varepsilon_{\mathrm{qre}})}{\eta\tau}
        \right),
    \]
    where the logarithmic notation absorbs the additional logarithm generated by the polynomial factor $(t+q_h+2)^{2q_h}$.

    Now by the definition of $G_h$, we have that
    \[
        G_h = O(N^4H^2 C_h^2/\tau).
    \]
    Because there exists a constant $C$ such that $C_h \leq (CN^3H\log A/\tau)^{H-h}$, we have that
    \[
        G_h \leq N^4H^2(CN^3H\log A/\tau)^{H-h}\precsim (CN^4H^2\log A/\tau)^{H-h}.
    \]
    Hence, \[\log G_h = \widetilde O(H)\]

    
    Substituting this bound, $q_h\leq H$,
    and $\eta^{-1}=\Theta(N^2H)$ yields the stated uniform iteration
    bound with $\log(e/\varepsilon_{\mathrm{qre}})$ for
    $\varepsilon_{\mathrm{qre}}\in(0,1]$. This completes the proof.
\end{proof}

\section{Structural properties of zero-sum networked separable Markov games}
\label{sec:zero-sum-properties}

We collect the structural identities that are repeatedly used in the proof.

\begin{lemma}[Zero-sum value and $Q$-function identities]
    \label{sumzero}
    For every Markov policy profile $\pi=(\pi_1,\ldots,\pi_N)$, every $h\in[H]$, $s\in\gS$, and $\va\in\gA$,
    \[
        \sum_{i=1}^{N}V_{h,i}^{\pi}(s)=0,
        \qquad
        \sum_{i=1}^{N}Q_{h,i}^{\pi}(s,\va)=0.
    \]
    Consequently, for every $h\in[H]$ and $s\in\gS$, the decomposed payoff collection
    \[
        \left(Q_{h,i,j}^{\pi}(s,\cdot,\cdot)\right)_{(i,j)\in\gE_Q}
    \]
    defines a zero-sum networked separable normal-form game.

    Moreover, the OMWU estimates satisfy
    \[
        \sum_{i=1}^{N}Q_{h,i}^{t}(s,\va)=0
        \qquad
        \text{for all }t\geq0,
    \]
    and, whenever $V^t$ has been computed,
    \[
        \sum_{i=1}^{N}V_{h,i}^{t}(s)=0.
    \]
\end{lemma}

\begin{proof}
    Fix a policy profile $\pi$. We prove the policy identities by backward induction on $h$. At $h=H+1$, $V_{H+1,i}^{\pi}\equiv0$, so the value identity is immediate. If $\sum_iV_{h+1,i}^{\pi}\equiv0$, then for every $(s,\va)$,
    \[
    \begin{aligned}
        \sum_{i=1}^{N}Q_{h,i}^{\pi}(s,\va)
        &=
        \sum_{i=1}^{N}r_{h,i}(s,\va)
        +
        \sP_h(\cdot|s,\va)\cdot
        \sum_{i=1}^{N}V_{h+1,i}^{\pi}(\cdot)
        =0,
    \end{aligned}
    \]
    where the first term is zero by the stagewise zero-sum assumption. Taking expectation of the last display under the common product distribution $\pi_h(s)$ gives $\sum_iV_{h,i}^{\pi}(s)=0$. This completes the induction. Since $Q_{h,i}^{\pi}=\sum_{j\in\gN_i}Q_{h,i,j}^{\pi}$, the decomposed collection is zero-sum in the networked separable sense.

    For OMWU, use the common nonnegative edge components
    $\sK_{h,i,j}$, whose sum over $j\in\gN_i$ is $\sP_h$ in both
    controller cases. In the empty-controller case this follows from
    $\sK_{h,i,j}=\sP_{h,o}/|\gN_i|$.
    Set $Q_{h,i}^{t}:=\sum_{j\in\gN_i}Q_{h,i,j}^{t}$. At $t=0$, $Q_{h,i}^0=r_{h,i}$, so $\sum_iQ_{h,i}^0=0$. Suppose that $\sum_iQ_{h,i}^{t}=0$. During the computation of the $(t+1)$-st iterate, argue backward in $h$. At $h=H$, the update gives $Q_{H,i}^{t+1}=(1-\eta\tau)Q_{H,i}^{t}+\eta\tau r_{H,i}$, hence $\sum_iQ_{H,i}^{t+1}=0$, and the value update gives $\sum_iV_{H,i}^{t+1}=0$. If $\sum_iV_{h+1,i}^{t+1}=0$, then summing the separated $Q$-update over $j\in\gN_i$ gives
    \[
        Q_{h,i}^{t+1}
        =
        (1-\eta\tau)Q_{h,i}^{t}
        +
        \eta\tau\left(r_{h,i}+\sP_hV_{h+1,i}^{t+1}\right),
    \]
    and therefore
    \[
        \sum_{i=1}^{N}Q_{h,i}^{t+1}
        =
        (1-\eta\tau)\sum_{i=1}^{N}Q_{h,i}^{t}
        +
        \eta\tau\sum_{i=1}^{N}r_{h,i}
        +
        \eta\tau\sP_h\sum_{i=1}^{N}V_{h+1,i}^{t+1}
        =0.
    \]
    Taking expectation under $\bar\pi_h^{t+1}(s)$ in the value update yields $\sum_iV_{h,i}^{t+1}(s)=0$. This proves the estimate identities by induction over $t$ and backward induction over $h$.
\end{proof}

\begin{lemma}[Zero-sum networked separable identities]
    \label{sumzero2}
    Let $(u_i)_{i\in[N]}$ be the multilinear utilities of a zero-sum networked separable normal-form game. For any two product strategy profiles $\pi$ and $\pi'$, it holds that
    \[
        \sum_{i=1}^{N}
        \left[
            u_i(\pi_i,\pi'_{-i})
            +
            u_i(\pi'_i,\pi_{-i})
        \right]
        =0.
    \]
    Consequently, extending $u_i$ multilinearly to signed differences,
    \[
        \sum_{i=1}^{N}
        u_i(\pi_i-\pi'_i,\pi_{-i}-\pi'_{-i})
        =0.
    \]
\end{lemma}

\begin{proof}
    Let
    \[
        U(\rho):=\sum_{i=1}^{N}u_i(\rho).
    \]
    Since the game is zero-sum, $U(\rho)=0$ for every product strategy profile $\rho$. For $\theta\in[0,1]$, define the product profile
    \[
        \rho^\theta_i:=(1-\theta)\pi_i+\theta\pi'_i,
        \qquad i\in[N].
    \]
    The networked separable structure makes $U(\rho^\theta)$ a polynomial of degree at most two in $\theta$. Expanding by bilinearity gives
    \[
    \begin{aligned}
        U(\rho^\theta)
        &=(1-\theta)^2U(\pi)+\theta^2U(\pi')
        \\
        &\quad+
        \theta(1-\theta)
        \sum_{i=1}^{N}
        \left[
            u_i(\pi_i,\pi'_{-i})
            +
            u_i(\pi'_i,\pi_{-i})
        \right].
    \end{aligned}
    \]
    The left-hand side, $U(\pi)$, and $U(\pi')$ are all zero. Taking any $\theta\in(0,1)$ proves the first identity.

    The second identity follows by multilinearity:
    \[
    \begin{aligned}
        &\sum_i u_i(\pi_i-\pi'_i,\pi_{-i}-\pi'_{-i})
        \\
        &=
        \sum_i u_i(\pi)
        +
        \sum_i u_i(\pi')
        -
        \sum_i\left[u_i(\pi_i,\pi'_{-i})+u_i(\pi'_i,\pi_{-i})\right]
        =0.
    \end{aligned}
    \]
\end{proof}

\section{Technical lemmas}
\label{sec:technical-lemmas}

\begin{lemma}[Pinsker's inequality]
\label{lem:pinsker}
Let $\mathcal A$ be a finite set, and let $p,q\in\Delta(\mathcal A)$. Then
\[
    \frac{1}{2}\|p-q\|_1^2
    \leq
    \KL(p\|q),
\]
where
\[
    \KL(p\|q)
    :=
    \sum_{a\in\mathcal A}
    p(a)\log\frac{p(a)}{q(a)}
\]
with the convention that zero-mass terms contribute zero and
$\KL(p\|q)=+\infty$ if there exists $a\in\mathcal A$ such that
$p(a)>0$ and $q(a)=0$. Equivalently,
\[
    \|p-q\|_1
    \leq
    \sqrt{2\KL(p\|q)} .
\]
\end{lemma}

\begin{proof}
If $\KL(p\|q)=+\infty$, the claim is immediate. Hence assume that $p$ is absolutely continuous with respect to $q$. If $p=q$, both sides are zero, so suppose $p\ne q$.

Let
\[
    \delta
    :=
    \frac{1}{2}\|p-q\|_1
\]
be the total variation distance, and define
\[
    A
    :=
    \{a\in\mathcal A: p(a)\ge q(a)\}.
\]
Then
\[
    p(A)-q(A)=\delta .
\]
By the log-sum inequality, coarse-graining the distributions into the two events
$A$ and $A^c$ cannot increase KL divergence, so
\[
\begin{aligned}
    \KL(p\|q)
    &\ge
    p(A)\log\frac{p(A)}{q(A)}
    +
    p(A^c)\log\frac{p(A^c)}{q(A^c)} .
\end{aligned}
\]
Write $u=p(A)$ and $v=q(A)$. Since $p\ne q$, we have $\delta>0$.
Absolute continuity implies $v>0$ because $u>0$, while
$u=v+\delta\leq1$ implies $v<1$. Thus $0<v<1$ and $u\in(v,1]$.
The right-hand side is the binary KL divergence
\[
    d(u\|v)
    :=
    u\log\frac{u}{v}
    +
    (1-u)\log\frac{1-u}{1-v}.
\]
We now show that
\[
    d(u\|v)\ge 2(u-v)^2 .
\]
Fix $v\in(0,1)$ and define
\[
    f(u):=d(u\|v)-2(u-v)^2 .
\]
Then
\[
    f(v)=0,\qquad f'(v)=0,
\]
and
\[
    f''(u)
    =
    \frac{1}{u}
    +
    \frac{1}{1-u}
    -
    4
    =
    \frac{1}{u(1-u)}-4
    \ge 0,
\]
for $u\in(0,1)$, because $u(1-u)\le 1/4$.
Hence $f$ is convex and attains its minimum at $u=v$.
By continuity the same lower bound holds at $u=1$, so
\[
    d(u\|v)\ge 2(u-v)^2=2\delta^2.
\]
Since $\delta=\frac12\|p-q\|_1$, we obtain
\[
    \KL(p\|q)
    \ge
    2\delta^2
    =
    \frac12\|p-q\|_1^2.
\]
This proves the lemma.
\end{proof}

\begin{lemma}[Multilinear policy-perturbation bound]
    \label{lem:multilinear-tv}
    Let $F$ be the multilinear extension of an $N$-player payoff tensor with $\|F\|_\infty\leq NH$ on pure action profiles. For any two product distributions $\rho=(\rho_1,\ldots,\rho_N)$ and $\sigma=(\sigma_1,\ldots,\sigma_N)$,
    \[
        |F(\rho)-F(\sigma)|
        \leq
        NH\sum_{m=1}^{N}\|\rho_m-\sigma_m\|_1.
    \]
    Consequently,
    \[
        |F(\rho)-F(\sigma)|
        \leq
        NH\sqrt{2N\KL(\rho||\sigma)}
    \]
    whenever $\KL(\rho||\sigma)<\infty$, and the same bound holds with $\KL(\sigma||\rho)$ in place of $\KL(\rho||\sigma)$ whenever that quantity is finite.
\end{lemma}

\begin{proof}
    Let $\rho^0:=\rho$ and $\rho^N:=\sigma$. For $m=1,\ldots,N$, let $\rho^m$ be the product distribution obtained from $\rho^{m-1}$ by replacing only the $m$-th marginal by $\sigma_m$. By telescoping,
    \[
        |F(\rho)-F(\sigma)|
        \leq
        \sum_{m=1}^{N}|F(\rho^{m-1})-F(\rho^m)|.
    \]
    For a fixed $m$, the difference is an inner product between $\rho_m-\sigma_m$ and a vector whose entries are conditional expectations of $F$ and hence have absolute value at most $NH$. Thus
    \[
        |F(\rho^{m-1})-F(\rho^m)|
        \leq
        NH\|\rho_m-\sigma_m\|_1.
    \]
    Summing over $m$ proves the total-variation bound. The KL versions follow from Cauchy--Schwarz and Pinsker's inequality in Lemma~\ref{lem:pinsker}:
    \[
        \sum_{m=1}^{N}\|\rho_m-\sigma_m\|_1
        \leq
        \sqrt{N\sum_{m=1}^{N}\|\rho_m-\sigma_m\|_1^2}
        \leq
        \sqrt{2N\sum_{m=1}^{N}\KL(\rho_m||\sigma_m)}.
    \]
    Reversing the order of the arguments gives the analogous bound with $\KL(\sigma||\rho)$.
\end{proof}

\begin{lemma}[Policy perturbation bound for Markov-game $Q$-functions]
    \label{lem:q-policy-perturb}
    \label{lem34}
    For any two Markov policy profiles $\mu$ and $\pi$, every $h\in[H]$, and every $i\in[N]$,
    \[
        \left\|Q_{h,i}^{\mu}-Q_{h,i}^{\pi}\right\|_\infty
        \leq
        NH
        \sum_{\ell=h+1}^{H}
        \max_{s\in\gS}
        \sum_{m=1}^{N}
        \left\|
            \mu_{\ell,m}(\cdot|s)
            -
            \pi_{\ell,m}(\cdot|s)
        \right\|_1.
    \]
    In particular, with $\mu=\bar\pi^t$ and $\pi=\pi^{*,\tau}$,
    \[
        \left\|Q_{h,i}^{\bar\pi^t}-Q_{h,i}^{*,\tau}\right\|_\infty
        \leq
        NH
        \sum_{\ell=h+1}^{H}
        \sqrt{2N}\,\bar X_{\ell}^{t}.
    \]
\end{lemma}

\begin{proof}
    Define
    \[
        \Delta_{h,i}^Q:=\|Q_{h,i}^{\mu}-Q_{h,i}^{\pi}\|_\infty,
        \qquad
        \Delta_{h,i}^V:=\|V_{h,i}^{\mu}-V_{h,i}^{\pi}\|_\infty.
    \]
    Since the rewards are the same for both policy profiles,
    \[
        \Delta_{h,i}^Q
        \leq
        \Delta_{h+1,i}^V.
    \]
    For every $s\in\gS$,
    \[
    \begin{aligned}
        |V_{h,i}^{\mu}(s)-V_{h,i}^{\pi}(s)|
        &\leq
        |Q_{h,i}^{\mu}(s,\mu_h(s))-Q_{h,i}^{\pi}(s,\mu_h(s))|
        \\
        &\quad+
        |Q_{h,i}^{\pi}(s,\mu_h(s))-Q_{h,i}^{\pi}(s,\pi_h(s))|.
    \end{aligned}
    \]
    The first term is at most $\Delta_{h,i}^Q$. Since $\|Q_{h,i}^{\pi}\|_\infty\leq NH$, Lemma~\ref{lem:multilinear-tv} bounds the second term by
    \[
        NH
        \sum_{m=1}^{N}
        \left\|
            \mu_{h,m}(\cdot|s)
            -
            \pi_{h,m}(\cdot|s)
        \right\|_1.
    \]
    Hence
    \[
        \Delta_{h,i}^V
        \leq
        \Delta_{h,i}^Q
        +
        NH
        \max_{s\in\gS}
        \sum_{m=1}^{N}
        \left\|
            \mu_{h,m}(\cdot|s)
            -
            \pi_{h,m}(\cdot|s)
        \right\|_1.
    \]
    Combining this with $\Delta_{h,i}^Q\leq\Delta_{h+1,i}^V$ gives
    \[
        \Delta_{h,i}^Q
        \leq
        \Delta_{h+1,i}^Q
        +
        NH
        \max_{s\in\gS}
        \sum_{m=1}^{N}
        \left\|
            \mu_{h+1,m}(\cdot|s)
            -
            \pi_{h+1,m}(\cdot|s)
        \right\|_1.
    \]
    Since $Q_{H,i}^{\mu}=Q_{H,i}^{\pi}=r_{H,i}$, we have $\Delta_{H,i}^Q=0$. Iterating the preceding recursion backward proves the first claim.

    For the special case $\mu=\bar\pi^t$ and $\pi=\pi^{*,\tau}$, Lemma~\ref{lem:qre-induced-structure} identifies $Q^{*,\tau}$ with $Q^{\pi^{*,\tau}}$. Cauchy--Schwarz and Pinsker's inequality in Lemma~\ref{lem:pinsker} give, for every $\ell$ and $s$,
    \[
    \begin{aligned}
        \sum_{m=1}^{N}
        \left\|
            \bar\pi_{\ell,m}^{t}(\cdot|s)
            -
            \pi_{\ell,m}^{*,\tau}(\cdot|s)
        \right\|_1
        &\leq
        \sqrt{
            2N
            \KL\left(
                \pi_{\ell}^{*,\tau}(s)
                ||
                \bar\pi_{\ell}^{t}(s)
            \right)
        }
        \\
        &\leq
        \sqrt{2N}\,\bar X_{\ell}^{t}.
    \end{aligned}
    \]
    Substituting this into the first bound proves the second claim.
\end{proof}

\section{Supplementary numerical results}
\label{app:supplementary-experiments}

\begin{figure}[htbp]
    \centering
    \includegraphics[width=\linewidth]{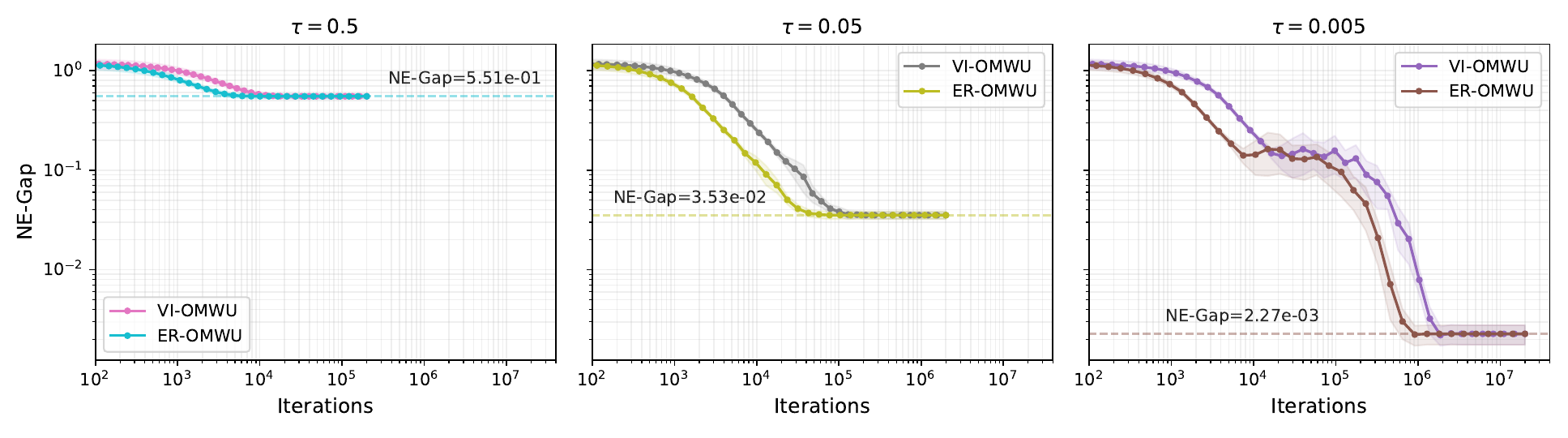}
    \caption{Convergence of the first-stage Nash-equilibrium gap for the
    complementary temperatures $\tau\in\{0.5,0.05,0.005\}$. Curves show means
    over 10 sampled game instances on log-log axes, with VI-OMWU plotted using
    the matched update-count convention.}
    \label{fig:negap-extra-tau}
\end{figure}

This supplementary figure reports the same comparison as \Cref{fig:negap} for
the remaining values in the temperature grid. The qualitative behavior is
consistent with the main panels: ER-OMWU reaches a terminal NE-gap comparable
to VI-OMWU under the matched update-count convention, and smaller $\tau$
corresponds to a lower terminal gap after a longer transient period.

\end{document}